%% file: tw-mat.tex
\documentclass[11pt,letter]{article}
\usepackage[margin=1in]{geometry}
\usepackage[latin2]{inputenc}
\usepackage[english]{babel}
\usepackage{amsmath}
\usepackage[all=normal,paragraphs=tight,bibliography=tight]{savetrees}
\usepackage{todonotes}
\usepackage{amsthm}
\usepackage{amssymb}
\usepackage{microtype}
\usepackage{xspace}
\usepackage{comment}
\usepackage{letltxmacro}
\usepackage{comment}
\usepackage{enumerate}
\usepackage{tikz}
\usetikzlibrary{calc,shapes}
\usepackage{url}
\usepackage{eufrak}
\usepackage{ogonek}
\let\chapter\section
\usepackage{algorithm2e}

\usepackage{thmtools}
\usepackage{thm-restate}

\usepackage{graphicx}
\usepackage{caption}
\usepackage{subcaption}
\usepackage{float}

\newtheorem{theorem}{Theorem}[section]
\newtheorem{lemma}[theorem]{Lemma}
\newtheorem{claim}[theorem]{Claim}
\newtheorem{corollary}[theorem]{Corollary}

\theoremstyle{definition}
\newtheorem{definition}[theorem]{Definition}

\newcommand{\retheorem}{theorem}

\def\cqedsymbol{\ifmmode$\lrcorner$\else{\unskip\nobreak\hfil
\penalty50\hskip1em\null\nobreak\hfil$\lrcorner$
\parfillskip=0pt\finalhyphendemerits=0\endgraf}\fi} 

\newcommand{\cqed}{\renewcommand{\qed}{\cqedsymbol}}

\newcommand{\executeiffilenewer}[3]{%
\ifnum\pdfstrcmp{\pdffilemoddate{#1}}%
{\pdffilemoddate{#2}}>0%
{\immediate\write18{#3}}\fi%
} 
\newcommand{\R}{\ensuremath{\mathbb{R}}}
\newcommand{\Oh}{\ensuremath{\mathcal{O}}}

\newcommand{\poly}{\mathrm{poly}}
\newcommand{\tw}{\ensuremath{\mathtt{tw}}\xspace}
\newcommand{\tpw}{\ensuremath{\mathtt{tpw}}\xspace}
\newcommand{\pw}{\ensuremath{\mathtt{pw}}\xspace}

\DeclareMathOperator{\rk}{rk}

\newcommand{\Ff}{\mathcal{F}}
\newcommand{\Tt}{\mathcal{T}}
\newcommand{\eps}{\varepsilon}
\newcommand{\F}{\mathbb{F}}
\newcommand\mleq{\preccurlyeq}
\def\E{\ensuremath{\mathcal{E}}\xspace}
\def\T{\Tt}
\newcommand{\treedecomp}{\Tt}

\title{Fully polynomial-time parameterized computations for graphs and matrices of low treewidth\thanks{D. Lokshtanov is supported by the BeHard grant under the recruitment programme of the Bergen Research Foundation.
The research of F. V. Fomin leading to these results has received funding from the European Research Council under the European Union's Seventh Framework Programme (FP/2007-2013) / ERC Grant Agreement n. 267959.
S. Saurabh is supported by PARAPPROX, ERC starting grant no. 306992. The research of Mi. Pilipczuk and M. Wrochna is supported by Polish National Science Centre grant DEC-2013/11/D/ST6/03073. During the work on these results, Micha\l{} Pilipczuk held a post-doc position at Warsaw Center of Mathematics and Computer Science, and was supported by the Foundation for Polish Science (FNP) via the START stipend programme.}}
\author{
  Fedor V. Fomin\thanks{
    Department of Informatics, University of Bergen, Norway, \texttt{fomin@ii.uib.no}.
  }
  \and
  Daniel Lokshtanov\thanks{
    Department of Informatics, University of Bergen, Norway, \texttt{daniello@ii.uib.no}.
  }
  \and
  Micha\l{} Pilipczuk\thanks{
    Institute of Informatics, University of Warsaw, Poland, \texttt{michal.pilipczuk@mimuw.edu.pl}.
  }
  \and 
  Saket Saurabh\thanks{
    Institute of Mathematical Sciences, India, \texttt{saket@imsc.res.in}, and
    Department of Informatics, University of Bergen, Norway, \texttt{Saket.Saurabh@ii.uib.no}.
  }
  \and
  Marcin Wrochna\thanks{
    Institute of Informatics, University of Warsaw, Poland, \texttt{m.wrochna@mimuw.edu.pl}.
  }
}

\date{}

\begin{document}

\begin{titlepage}
\def\thepage{}
\thispagestyle{empty}
\maketitle

\begin{abstract}
\input{abstract}
\end{abstract}
\end{titlepage}

\section{Introduction}\label{sec:intro}
\input{intro}

\section{Preliminaries}\label{sec:prelims}
\input{prelims}

\section{Approximating treewidth}\label{sec:approx}
\input{approx}

\section{Gaussian elimination}\label{sec:gaussian}
\input{gaussian}

\input{gaussian-tw}

\section{Maximum matching}\label{sec:matching}
\input{matching}

\section{Maximum flow}\label{sec:max-flow}
\input{max-flow}

\section{Conclusions}\label{sec:conclusions}
\input{conclusions}

\bibliographystyle{abbrv}
\bibliography{tw-mat}

\end{document}

%% file: abstract.tex
We investigate the complexity of several fundamental polynomial-time solvable problems on graphs and on matrices, when the given instance has low treewidth; in the case of matrices, we consider the treewidth of the graph formed by non-zero entries. In each of the considered cases, the best known algorithms working on general graphs run in polynomial time, however the exponent of the polynomial is large. Therefore, our main goal is to construct algorithms with running time of the form $\poly(k)\cdot n$ or $\poly(k)\cdot n\log n$, where $k$ is the width of the tree decomposition given on the input. Such procedures would outperform the best known algorithms for the considered problems already for moderate values of the treewidth, like $\Oh(n^{1/c})$ for some small constant $c$.

Our results include: 
\begin{enumerate}[--]
\item an algorithm for computing the determinant and the rank of an $n\times n$ matrix using $\Oh(k^3\cdot n)$ time and arithmetic operations;
\item an algorithm for solving a system of linear equations using $\Oh(k^3\cdot n)$ time and arithmetic operations;
\item an $\Oh(k^3\cdot n\log n)$-time randomized algorithm for finding the cardinality of a maximum matching in a graph; 
\item an $\Oh(k^4\cdot n\log^2 n)$-time randomized algorithm for constructing a maximum matching in a graph;
\item an $\Oh(k^2\cdot n\log n)$-time algorithm for finding a maximum vertex flow in a directed graph.
\end{enumerate}
Moreover, we give an approximation algorithm for treewidth with time complexity suited to the running times as above. Namely, the algorithm, when given a graph $G$ and integer $k$, runs in time $\Oh(k^7\cdot n\log n)$ and either correctly reports that the treewidth of $G$ is larger than $k$, or constructs a tree decomposition of $G$ of width $\Oh(k^2)$.

The above results stand in contrast with the recent work of Abboud et al.~[to appear at SODA 2016], which shows that the existence of algorithms with similar running times is unlikely for the problems of finding the diameter and the radius of a graph of low treewidth.

%% file: intro.tex
Treewidth is a fundamental graph parameter that measures how much the graph's structure resembles a tree. This resemblance is measured by how well we can decompose the graph in a treelike manner; the optimum width of such a {\em{tree decomposition}} is exactly the definition of treewidth. It has been known for a long time that many computational problems can be solved more efficiently on graphs of low treewidth, which is important from the point of view of applications, as such graphs do appear in practice. For instance, the control-flow graphs of programs in popular programming languages have constant treewidth~\cite{GustedtMT02,Thorup98}, whereas topologically-constrained graphs, like planar graphs or $H$-minor free graphs, have treewidth $\Oh(n^{1/2})$, where $n$ is the vertex count.

Perhaps the most important result about algorithms on low treewidth graphs is delivered by Courcelle's celebrated theorem, which asserts that every problem expressible in (the optimization variant of) {\em{Monadic Second Order logic}} admits an algorithm with running time $f(k)\cdot n$ on graphs of treewidth $k$, for some function $f$~\cite{ArnborgLS91,Courcelle90}. While there are no good upper bounds on the function $f$ in general, more precise estimates are known for specific problems of interest; in particular, $f$ is usually single-exponential for the most common NP-hard problems. The investigation of algorithms working on graphs of low treewidth is by now a well-established research direction, which is perhaps best captured by the methodology of {\em{parameterized complexity}}. This paradigm focuses on studying the complexity of computational problems using auxiliary measures of hardness, called {\em{parameters}}; treewidth is one of the most important examples of such parameters. We invite the reader to the relevant chapters in textbooks on parameterized complexity~\cite{platypus,DowneyF13,FlumGroheBook,niedermeier:book} for a broader discussion.

Of course, as long as the considered problem is NP-hard, we cannot expect to obtain an algorithm working on a graph of treewidth $k$ with running time of the form $f(k)\cdot n^c$ for some constant $c$ and a {\em{polynomial}} function $f$. However, let us consider any problem that can be solved in polynomial time, but for which no near-linear time algorithm is known; say, {\sc{Maximum Matching}}. There is a wide variety of algorithms for finding a maximum matching in a graph; however, their running times are far from linear. On the other hand, it is not hard to obtain an algorithm working on a graph given together with a tree decomposition of width $k$ that works in time $\Oh(3^k\cdot k^{\Oh(1)}\cdot n)$; this outperforms all the general-purpose algorithms for constant values of the treewidth. Is it possible to obtain an algorithm with running time $\Oh(k^{d}\cdot n)$ or $\Oh(k^{d}\cdot n\log n)$ for some constant $d$, implying a significant speed-up already for moderate values of treewidth like $k=\Oh(n^{1/3d})$? This question can be asked also for a number of other problems for which the known general-purpose polynomial-time algorithms have unsatisfactory running times.

In this paper we initiate a systematic study of the treewidth parameterization for fundamental problems that are solvable in polynomial time, but for which a lower exponent in the running time would be desirable. Examples of such problems include {\sc{maximum matching}}, {\sc{maximum flow}}, or various algebraic problems on matrices, like computing determinants or solving systems of linear equations. These are exactly the problems that we address in this work. The goal is to design an algorithm with running time of the form $\Oh(k^d\cdot p(n))$ for some constant $d$ and polynomial $p(n)$ that would be much smaller than the running time bound of the fastest known unparameterized algorithm. Mirroring the terminology of parameterized complexity, we will call such algorithms {\em{fully polynomial FPT}} (FPT stands for {\em{fixed-parameter tractable}}). Although several results of this kind are scattered throughout the literature~\cite{AkibaSK12,ChatterjeeL13,ChaudhuriZ98,ChaudhuriZ00,PlankenWK12}, mostly concerning shortest path problems, no systematic investigations have been made so far.

This direction fits into the general concept of ``FPT within P'' that was coined very recently by Giannopoulou et al.~\cite{GiannopoulouMN15} and by Abboud et al.~\cite{AbboudWW15}. In particular, Abboud et al.~\cite{AbboudWW15} made already the first step into investigating the treewidth parameterizations by considering the {\sc{Diameter}} and {\sc{Radius}} problems. They proved that both these problems can be solved in time $2^{\Oh(k\log k)}\cdot n^{1+o(1)}$ on graphs of treewidth $k$, but achieving running time of the form $2^{o(k)}\cdot n^{2-\eps}$ for any $\eps>0$ for {\sc{Diameter}} would already contradict the Strong Exponential Time Hypothesis ({\em{SETH}}) of Impagliazzo et al.~\cite{ImpagliazzoPZ01}; the same lower bound is also given for {\sc{Radius}}, but under a stronger assumption. This lays foundations for a lower bound methodology for the studied class of running times, and suggests that now we have all the right tools for a systematic study of the complexity of treewidth parameterizations of polynomial-time solvable problems.


\paragraph*{Our contribution.} The goal of this paper is to provide solid algorithmic foundations for the further study of fully polynomial parameterized algorithms on graphs and matrices of low treewidth. For this, we design such algorithms for several fundamental problems, which can then serve as vital subroutines in future investigations.


The vast majority of algorithms working on low treewidth graphs assume that a suitable tree decomposition is given on the input. For this reason, one of the fundamental problems is to compute such a decomposition efficiently. Computing treewidth exactly is NP-hard~\cite{ArnborgCP87}, however there is a wide variety of FPT and approximation algorithms for computing near-optimal tree decompositions. Unfortunately, none of the results known in the literature (which we describe in more details in the section on related work) suits the studied running times: either the dependence on the target width is exponential, or the polynomial factor is far larger than quasi-linear and cannot be easily improved with the same approach.
For this reason, we propose a new approximation algorithm for treewidth, suited for fully polynomial FPT algorithms.

\begin{restatable}{\retheorem}{restateapx}\label{thm:apx}
There exists an algorithm that, given a graph $G$ on $n$ vertices and a positive integer~$k$, in time $\Oh(k^7\cdot n\log n)$ either provides a tree decomposition of $G$ of width at most $\Oh(k^2)$, or correctly concludes that $\tw(G)\geq k$.
\end{restatable}

Thus, Theorem~\ref{thm:apx} can serve the same role for fully polynomial FPT algorithms parameterized by treewidth, as Bodlaender's linear-time algorithm for treewidth~\cite{Bodlaender96} serves for Courcelle's theorem~\cite{Courcelle90}: it can be used to remove the assumption that a suitable tree decomposition is given on the input, because such a decomposition can be approximated roughly within the same running time.

Next, we turn to algebraic problems on matrices. Given an $n\times m$ matrix $A$ over some field, we can construct a bipartite graph $G_A$ as follows: the vertices on the opposite sides of the bipartition correspond to rows and columns of $A$, respectively, and a row is adjacent to a column if and only if the entry on their intersection is non-zero in $A$. Then, we can investigate the complexity of computational problems when a tree decomposition of $G_A$ of (small) width $k$ is given on the input. As a graph on $n$ vertices and of treewidth $k$ has at most $kn$ edges, it follows that such matrices are sparse --- they contain only $\Oh(kn)$ non-zero entries. It is perhaps more convenient to think of them as edge-weighted bipartite graphs: we label each edge of $G_A$ with the element placed in the corresponding entry of the matrix. We assume that this is the form in which the matrix is given to the algorithm, and hence we can talk about subquadratic algorithms for such matrices.

Our main result here is a pivoting scheme that essentially enables us to perform Gaussian elimination on matrices of small treewidth. In particular, we are able to obtain useful factorizations of such matrices, which gives us information about the determinant and rank, and the possibility to solve linear equations efficiently.
We cannot expect to invert matrices in near-linear time, as even very simple matrices have inverses with $\Omega(n^2)$ entries (e.g. the square matrix with $M[j,i]=1$ for $i-j\in\{0,1\}$, $0$ elsewhere).
The following theorems gather the main corollaries of our results, whereas we refer to Sections~\ref{sec:prelims} and~\ref{sec:gaussian} for more details on the form of factorizations that we obtain.
Note that for square matrices, the same results can be applied to decompositions of the usual symmetric graph, as explained in Section~\ref{sec:prelims}.

\begin{restatable}{\retheorem}{restatepwdet}\label{thm:pw-det}
	Given an $n\times m$ matrix $M$ over a field $\F$ and a path or tree-partition decomposition of its bipartite graph $G_M$ of width $k$,
	Gaussian elimination on $M$ can be performed using $\Oh(k^2\cdot (n+m))$ field operations and time.
	In particular, the rank, determinant, a maximal nonsingular submatrix and a $PLUQ$-factorization can be computed in this time.
	Furthermore, for every $r\in \F^n$, the system of linear equations $M x = r$ can be solved in $\Oh(k\cdot (n+m))$ additional field operations and time.
\end{restatable}

\begin{restatable}{\retheorem}{restatetwdet}\label{thm:tw-det}
	Given an $n\times m$ matrix $M$ over a field $\F$ and a tree-decomposition of its bipartite graph $G_M$ of width $k$,
	we can calculate the rank, determinant and a generalized $LU$-factorization of $M$ in $\Oh(k^3\cdot (n+m))$ field operations and time.
	Furthermore, for every $r\in \F^n$, the system of linear equations $M x = r$ can be solved in $\Oh(k^2 \cdot(n+m))$ additional field operations and time.
\end{restatable}

Our algorithms work more efficiently for parameters pathwidth and tree-partition width (see precise definitions in Section~\ref{sec:prelims}), which can be larger than treewidth. The reason is the pivoting scheme underlying Theorems~\ref{thm:pw-det} and~\ref{thm:tw-det} works perfectly for path and tree-partition decompositions, whereas for standard tree decomposition the scheme can possibly create a lot of new non-zero entries in the matrix. However, we show how to reduce the case of tree decompositions to tree-partition decompositions by adjusting the idea of matrix sparsification for nested dissection of Alon and Yuster~\cite{AlonY13} to the setting of tree decompositions. Unfortunately, this reduction incurs an additional $k$ factor in the running times of our algorithms, and we obtain a less robust factorization. 

Observe that one can also use the known inequality $\pw(G)\leq \tw(G)\cdot \log_2 n$ to reduce the treewidth case to the pathwidth case. This trades the additional factor $k$ for a factor $\log^2 n$; depending on the actual value of $k$, this might be beneficial for the overall running time.

Note that Theorems~\ref{thm:pw-det} and~\ref{thm:tw-det} work over any field $\F$. Hence, we can use them to develop an algebraic algorithm for the maximum matching problem, using the classic approach via the Tutte matrix.
This requires working in a field $\F$ (say, $\F_p$) of polynomial size, and hence the complexity of performing arithmetic operations in this field depends on the computation model.
Below we count all such operations as constant time, and elaborate on this issue in Section~\ref{sec:matching}.


\begin{restatable}{\retheorem}{restatesize}\label{thm:matching-size}
There exists an algorithm that, given a graph $G$ together with its tree decomposition of width at most $k$, uses $\Oh(k^3\cdot n)$ time and field operations and computes the size of a maximum matching in $G$. The algorithm is randomized with one-sided error: it is correct with probability at least $1-\frac{1}{n^c}$ for an arbitrarily chosen constant $c$, and in the case of an error it reports a suboptimal value.
\end{restatable}

Theorem~\ref{thm:matching-size} only provides the size of a maximum matching; to construct the matching itself, we need some more work.

\begin{restatable}{\retheorem}{restatereconst}\label{thm:matching-reconstruct}
There exists an algorithm that, given a graph $G$ together with its tree decomposition of width at most $k$, uses $\Oh(k^4\cdot n\log n)$ time and field operations and computes a maximum matching in $G$. The algorithm is randomized with one-sided error: it is correct with probability at least $1-\frac{1}{n^c}$ for an arbitrarily chosen constant $c$, and in the case of an error it reports a failure or a suboptimal matching.
\end{restatable}

We remark that our algebraic approach is tailored to unweighted graphs, and cannot be easily extended to the weighted setting.

Finally, we turn our attention to the maximum flow problem. We prove that for vertex-disjoint flows we can also design a fully polynomial FPT algorithm with near-linear running time dependence on the size of the input. The algorithm works even on directed graphs (given a tree decomposition of the underlying undirected graph), but only in the unweighted setting (i.e., with unit vertex capacities, which boils down to finding vertex-disjoint paths).

\begin{restatable}{\retheorem}{restatemaxflow}\label{thm:max-flow}
There exists an algorithm that given an unweighted directed graph $G$ on $n$ vertices, distinct terminals $s,t\in V(G)$ with $(s,t)\notin E(G)$, and a tree decomposition of $G$ of width at most $k$, works in time $\Oh(k^2\cdot n\log n)$ and computes a maximum $(s,t)$-vertex flow together with a minimum $(s,t)$-vertex cut in $G$.  
\end{restatable}

Theorem~\ref{thm:max-flow} states the result only for single-source and single-sink flows, but it is easy to reduce other variants, like $(S,T)$-flows, to this setting. Note that in particular, Theorem~\ref{thm:max-flow} provides an algorithm for the maximum matching problem in bipartite graphs that is faster than the general one from Theorem~\ref{thm:matching-reconstruct}: one just needs to add a new source $s$ and a new sink $t$ to the graph, and make $s$ and $t$ fully adjacent to the opposite sides of the bipartition. 

\paragraph*{Related work on polynomial-time algorithms on small treewidth graphs.} The reachability and shortest paths problems on low treewidth graphs have received considerable attention in the literature, especially from the point of view of data structures~\cite{AkibaSK12,ChaudhuriZ98,ChaudhuriZ00,ChatterjeeL13,PlankenWK12}. In these works, the running time dependence on treewidth is either exponential or polynomial, which often leads to interesting trade-off questions. As far as computation of maximum flows is concerned, we are aware only of the work of Hagerup et al. on multicommodity flows~\cite{HagerupKNR98}; however, their approach inevitably leads to exponential running time dependence on the treewidth. The work of Hagerup et al.~\cite{HagerupKNR98} was later used by Chambers and Eppstein~\cite{ChambersE13} for the maximum flow problem in one-crossing-minor-free graphs; unfortunately, the exponential dependency on the size of the excluded minor persists.

\paragraph*{Related work on approximating treewidth.} Computing treewidth exactly is NP-hard~\cite{ArnborgCP87}, and moreover there is no constant-factor approximation for treewidth unless the Small Set Expansion Hypothesis fails~\cite{WuAPL14}. However, when we allow the algorithm to run in FPT time when parameterized by the target width, then there is a wide variety of exact and approximation algorithms. Perhaps the best known are: the $4$-approximation algorithm in $2^{\Oh(k)}\cdot n^2$ time of Robertson and Seymour~\cite{gm13} (see~\cite{platypus,kleinberg-tardos} for an exposition of this algorithm) and the linear-time exact algorithm of Bodlaender with running time $k^{\Oh(k^3)}\cdot n$~\cite{Bodlaender96}. Recently, Bodlaender et al.~\cite{BodlaenderDDFLP13} obtained a 3-approximation in time $2^{\Oh(k)}\cdot n\log n$ and a 5-approximation in time $2^{\Oh(k)}\cdot n$. Essentially all the known approximation algorithms for treewidth follow the approach of Robertson and Seymour~\cite{gm13}, which is based on recursively decomposing subgraphs by breaking them using balanced separators.

As far as polynomial-time approximation algorithms are concerned, the best known algorithm is due to Feige et al.~\cite{FeigeHL08} and it achieves approximation factor $\Oh(\sqrt{\log OPT})$ in polynomial time. Unfortunately, the running time is far from linear due to the use of semi-definite programming for the crucial subroutine of finding balanced separators; this is also the case in previous works~\cite{LeightonR99,Amir10}, which are based on linear programming as well. 

For this reason, in the proof of Theorem~\ref{thm:apx} we develop a purely combinatorial $\Oh(OPT)$-factor approximation algorithm for finding balanced separators. This algorithm is based on the techniques of Feige and Mahdian~\cite{FeigeM06}, which are basic enough so that they can be implemented within the required running time. The new approximation algorithm for balanced separators is then combined with a trick of Reed~\cite{Reed92}. Essentially, the original algorithm of Robertson and Seymour~\cite{gm13} only breaks the (small) interface between the subgraph being decomposed and the rest of the graph, which may result in $\Omega(n)$ recursion depth. Reed~\cite{Reed92} observed that one can add an additional step of breaking the whole subgraph in a balanced way, which reduces the recursion depth to logarithmic and results in improving the running time dependence on the input size from $\Oh(n^2)$ to $\Oh(n\log n)$, at the cost of solving a more general (and usually more difficult) subproblem concerning balanced separators. Fortunately, our new approximation algorithm for balanced separators is flexible enough to solve this more general problem as well, so we arrive at $\Oh(n\log n)$ running time dependence on $n$.

\paragraph*{Related work on matrix computations.}
Solving systems of linear equations and computing the determinant and rank of a matrix are ubiquitous, thoroughly explored topics in computer science, with a variety of well-known applications.
Since sparse matrices often arise both in theory and in practice, the possibility (and often necessity) of exploiting their sparsity has been deeply studied as well.
Here we consider matrices as sparse when their non-zero entries are not only few (that is, $o(n^2)$), but furthermore they are structured in a way that could potentially be exploited using graph-theoretical techniques.
The two best known classical approaches in this direction are standard Gaussian elimination with a \emph{perfect elimination ordering}, and \emph{nested dissection}. 

A common assumption in both approaches is that throughout the execution of an algorithm, no accidental cancellation occurs -- that is, except for situations guaranteed and required by the algorithm, arithmetic operations never change a non-zero value to a zero.
This can be assumed in some settings, such as when the input matrix is positive definite.
Otherwise, as soon accidental zeroes occur (for example, simply starting with a zero diagonal) some circumvention is required by finding a different pivot than originally planned.
In practice, especially when working over real-valued matrices, one may expect this not to extend resource usage too much, but when working over finite fields it is clear that this assumption cannot be used to justify any resource bounds.
Surprisingly, we are not aware of any work bounding the worst-case running time of an algorithm for the determinant of a matrix of small treewidth (or pathwidth) without this assumption.
This may in part be explained by the fact that a better understanding of sparseness in graph theory and the rise of treewidth in popularity came after the classical work on sparse matrices, and by the reliance on heuristics in practice.

\emph{Perfect elimination ordering}, generally speaking, refers to an ordering of rows and columns of a matrix such that Gaussian elimination introduces no \emph{fill-in} -- entries in the matrix where a zero entry becomes non-zero.
A seminal result of Parter~\cite{parter61} and Rose~\cite{Rose1970597} says that such an ordering exists if and only if the (symmetric) graph of the matrix is chordal (or \emph{triangulated}, that is, every cycle with more than three edges has a chord).
This assumes 
no accidental cancellation occurs.
Hence to minimize space usage, one would search for a minimum completion to a chordal graph (the \textsc{Minimum Fill-in} problem), while to put a guarantee on the time spent on eliminating, one could demand a chordal completion with small cliques, which is equivalent to the graph of the original matrix having small treewidth.
Radhakrishnan et al.~\cite{RadhakrishnanHS92} use this approach to give an $\Oh(k^2 n)$ algorithm for solving systems of linear equation defined by matrices of treewidth $k$, assuming no accidental cancellation.

To lift this assumption one has to consider arbitrary pivoting and the bipartite graph of a matrix instead (with separate vertices for each row and each column), which also allows the study of non-symmetric, non-square matrices.
A $\Gamma$-free ordering is an ordering of rows and columns of a matrix such that no $\left(\begin{smallmatrix}\star&\star\\\star&0\end{smallmatrix}\right)$ submatrix occurs -- it can be seen that such an ordering allows to perform Gaussian elimination with no fill-in ($\star$ represents a non-zero entry).
This corresponds to a \emph{strong ordering} of the bipartite matrix $G_M$ of the matrix -- an ordering $\mleq$ of vertices such that for all vertices $i,j,k,\ell$, if $i\mleq \ell$, $j\mleq k$, and $ji, ki, j\ell$ are edges, then $k\ell$ must be an edge too.
A bipartite graph is known to be \emph{chordal bipartite} graph (defined as a bipartite graph with no chordless cycles strictly longer than 4 -- note it need not be chordal) if and only if it admits a strong ordering (see e.g.~\cite{Dragan00}).
Golumbic and Goss~\cite{GolumbicG78} first related chordal bipartite graphs to Gaussian elimination with no fill-in, but assuming no accidental cancellation.
Bakonyi and Bono~\cite{bakonyiB97} showed that when an accidental cancellation occurs and a pivot cannot be used, a different pivot can always be found.
However, they do not consider the running time needed for finding the pivot, nor the number of arithmetic operations performed.

\emph{Nested dissection} is a Divide\&Conquer approach introduced by Lipton, Tarjan and Rose~\cite{liptonTR79} to solve a system of linear equations whose matrix is symmetric positive definite and whose graph admits a certain separator structure.
Intuitively, a \emph{weak separator tree} for a graph gives a small separator of the graph that partitions its vertices into two balanced parts, which after removing the separator, are recursively partitioned in the same way.
In work related to nested dissection, a \emph{small} separator means one of size $\Oh(n^\gamma)$, where $n$ is the number of remaining vertices of the graph and $\gamma<1$ is a constant ($\gamma=\frac{1}{2}$ for planar and $H$-minor-free graphs).
Thus an algorithm needs to handle a logarithmic number of separators whose total size is a geometric series bounded again by $\Oh(n^\gamma)$.
In modern graph-theoretic language, this most closely corresponds to a (balanced, binary) tree-depth decomposition of depth $\Oh(n^\gamma)$.

To use nested dissection for matrices $A$ that are not positive definite (so without assuming no accidental cancellation), Mucha and Sankowski~\cite{MuchaS06} used it on $A A^{T}$ instead, carefully recovering some properties of $A$ afterwards.
In order to guarantee a good separator structure for $A A^{T}$, however, they first need to decrease the degree of the graph of $A$ by an approach called \emph{vertex splitting}, introduced by Wilson~\cite{Wilson97}.
Vertex splitting is the operation of replacing a vertex $v$ with a path on three vertices $v',w,v''$ and replacing each incident edge $uv$ with either $uv'$ or $uv''$.
It is easy to see that this operation preserves the number of perfect matchings, for example.
The operation applied to the graph of a matrix can in fact be performed on the matrix, preserving its determinant too. 
By repeatedly splitting a vertex, we can transform it, together with incident edges, into a tree of degree bounded by 3.
Choosing an appropriate partition of the incident edges, the structure of the graph can be preserved; for example, the knowledge of a planar embedding can be used to stay in the class of planar graphs.
This allowed Mucha and Sankowski~\cite{MuchaS06} to find maximum matchings via Gaussian elimination in planar graphs in $\Oh(n^{\omega/2})$ time, where $\omega<2.38$ is the exponent of the best known matrix multiplication algorithm.
Yuster and Zwick~\cite{YusterZ07} showed that the weak separator tree structure can be preserved too, which allowed them to extend this result to $H$-minor-free graphs.

These methods were further extended by Alon and Yuster~\cite{AlonY13} to use vertex splitting and nested dissection on $A A^{T}$ for solving arbitrary systems of linear equations over any field, for matrices whose graphs admit a weak separator tree structure.
If the separators are of size $\Oh(n^\beta)$ and can be efficiently found, the algorithm works in $\Oh(n^{\omega \beta})$ time.
However, it is randomized and very involved, in particular requiring arithmetic computations in field extensions of polynomial size.
A careful translation of their proofs to tree decompositions could only give an $\Oh(k^5 \cdot n \log^3 n)$ randomized algorithm for matrices of treewidth $k$ (that is, $\Oh(n' k'^2)$~\cite{RadhakrishnanHS92} where $n'=nk$ and $k'=k^2 \log n$ after vertex splitting, with the recursion in~\cite{AlonY13} giving an additional $\log n$ factor).

Our approach differs in that for matrices with path or tree-partition decompositions of small width we show that a strong ordering respecting the decomposition can be easily found, and standard Gaussian elimination is enough, as long as the ordering is properly used when pivoting.
For matrices with small treewidth this does not seem possible (an apparent obstacle here is that not all chordal graphs have strong orderings).
However, a variant of the vertex splitting technique guided with a tree-decomposition allows us to simply (in particular, deterministically) reduce to the tree-partition case (instead of considering $A A^{T}$).

\paragraph*{Related work on maximum matchings.}
The existence of a perfect matching in a graph can be tested by calculating the determinant of the Tutte matrix~\cite{Tutte}.
Lov\'asz~\cite{Lovasz79} showed that the size of a maximum matching can be found by computing the rank, while Mucha and Sankowski~\cite{MuchaS04,MuchaS06} gave a randomized algorithm for extracting a maximum matching: in $\Oh(n^\omega)$ time for general graphs and $\Oh(n^{\omega/2})$ for planar graphs. The results on general graphs were later simplified by~\cite{Harvey09}.
Before that, Edmonds~\cite{edmonds1965paths} gave the first polynomial time algorithm, then bested by combinatorial algorithms of
Micali and Vazirani~\cite{RabinV89,vazirani2014proof}, Blum~\cite{Blum90}, and Gabow and Tarjan~\cite{GabowT91}, each running in $\Oh(m\sqrt{n})$ time.
Recently M\k{a}dry~\cite{Madry13} gave an $\Oh(m^{10/7})$ algorithm for the unweighted bipartite case.

\paragraph*{Related work on maximum flows.}
The maximum flow problem is a classic subject with a long and rich literature. Starting with the first algorithm of Ford and Fulkerson~\cite{Ford-Fulkerson_algo}, which works in time $\Oh(F\cdot (n+m))$ for integer capacities, where $F$ is the maximum size of the flow, a long chain of improvements and generalizations was proposed throughout the years; see e.g.~\cite{1970:din,EdmondsK72,EvenT75,GoldbergR98,GoldbergT88,Karzanov73,KingRT94,MalhotraKM78,Madry13,Orlin13}. The running times of these algorithms vary depending on the variants they solve, but all of them are far larger than linear. In particular, the fastest known algorithm in the directed unit-weight setting, which is the case considered in this work, is due to M\k{a}dry~\cite{Madry13} and works in time $\Oh(m^{10/7})$. For this reason, recently there was a line of work on finding near-linear $(1+\eps)$-approximation algorithms for the maximum flow problem~\cite{ChristianoKMST11,KelnerLOS14,LeeRS13,Sherman13}, culminating in a $(1+\eps)$-approximation algorithm working in undirected graphs in time $\Oh(\eps^{-2}\cdot m^{1+o(1)})$, proposed independently by Sherman~\cite{Sherman13} and by Kelner et al.~\cite{KelnerLOS14}.

\paragraph*{Outline.} In Section~\ref{sec:prelims} we establish notation and recall basic facts about matrices, flows, and tree-like decompositions of graphs. Section~\ref{sec:approx} is devoted to the approximation algorithm for treewidth, i.e., Theorem~\ref{thm:apx}. In Section~\ref{sec:gaussian} we give our results for problems on matrices of low width, and in particular we prove Theorems~\ref{thm:pw-det} and~\ref{thm:tw-det}. In Section~\ref{sec:matching} we apply these results to the maximum matching problem, proving Theorems~\ref{thm:matching-size} and~\ref{thm:matching-reconstruct}. Section~\ref{sec:max-flow} is focused on the maximum vertex flow problem and contains a proof of Theorem~\ref{thm:max-flow}. Finally, in Section~\ref{sec:conclusions} we gather short concluding remarks and state a number of open problems stemming from our work.

%% file: prelims.tex
\newcommand{\meas}{\mu}
\newcommand{\cc}{\texttt{cc}}
\newcommand{\Cc}{\mathcal{C}}

\paragraph*{Notation.} We use standard graph notation; cf.~\cite{platypus}. All the graphs considered in this paper are simple, i.e., they have no loops or multiple edges connecting the same endpoints. For a graph $G$ and $X\subseteq V(G)$, by $N_G[X]$ we denote the {\em{closed neighborhood}} of $X$, i.e., all the vertices that are either in $X$ or are adjacent to vertices of $X$, and by $N_G(X)=N_G[X]\setminus X$ we denote the {\em{open neighborhood}} of $X$. When $G$ is clear from the context, we drop the subscript. For a path $P$, the {\em{internal vertices}} of $P$ are all the vertices traversed by $P$ apart from the endpoints. Paths $P$ and $Q$ are {\em{internally vertex-disjoint}} if the no internal vertex of $P$ is traversed by $Q$ and vice versa. 
The set of connected components of a graph $G$ is denoted by $\cc(G)$. 

By $G[X]$ we denote the subgraph of $G$ induced by $X$, and we define $G-X=G[V(G)\setminus X]$. Graph $H$ is a {\em{subgraph}} of $G$, denoted $H\subseteq G$, if $V(H)\subseteq V(G)$ and $E(H)\subseteq E(G)$. We say $H$ is a {\em{completion}} of $G$ if $V(H)=V(G)$ and $E(H)\supseteq E(G)$.

A {\em{$1$-subdivision}} of a graph $G$ is obtained from $G$ by taking every edge $uv\in E(G)$, and replacing it with a new vertex $w_{uv}$ and edges $uw_{uv}$ and $w_{uv}v$.

For a positive integer $q$, we denote $[q]=\{1,2,\ldots,q\}$.

\paragraph*{Matrices.} For an $n\times m$ matrix $M$, the entry at the intersection of the $j$th row and $i$th column is denoted as $M[j,i]$. For sets $X\subseteq [n]$ and $Y\subseteq [m]$, by $[M]_{X,Y}$  we denote the $|X|\times |Y|$ matrix formed by the entries of $M$ appearing on the intersections of rows of $X$ and columns of $Y$.

The \emph{symmetric graph} of an $n \times n$ matrix (i.e., square, but not necessarily symmetric) is the undirected graph with vertices $\{1,\dots,n\}$ and an edge between $i$ and $j$ whenever one of $M[i,j]$ and $M[j,i]$ is non-zero.
The \emph{bipartite graph} of an $n \times m$ matrix is a bipartite, undirected graph with vertices in $\{r_1,\dots,r_n\} \cup \{c_1, \dots, c_m\}$ and an edge between $r_i$ and $c_j$ whenever $M[i,j]\neq 0$.

In this work, for describing the structure of a matrix $M$, we use the bipartite graph exclusively and denote it $G_M$, as it allows to express our results for arbitrary (not necessarily square) matrices.
Note that any tree decomposition of the symmetric graph of a square matrix $M$ can be turned into a decomposition of $G_M$ of twice the width plus $1$, by putting both the $i$-th row and $i$-th column in the same bag where index $i$ was.

A matrix is in (non-reduced) \emph{row-echelon form} if all zero rows (with only zero entries) are below all non-zero rows, and the leftmost non-zero coefficient of each row is strictly to the right of the leftmost non-zero coefficients of rows above it.
In particular, there are no non-zero entries below the diagonal.
We define column-echelon form analogously.
A {\em{$PLUQ$-factorization}} of an $n\times m$ matrix $M$ (also known as an LU-factorization with full pivoting) is a quadruple of matrices where: $P$ is a permutation $n\times n$ matrix, $L$ is an $n\times n$ matrix in column-echelon form with ones on the diagonal, $U$ is an $n\times m$ matrix in row-echelon form, $Q$ is a permutation $m\times m$ matrix, and $M=PLUQ$.
A {\em{generalized $LU$-factorization}} of $M$ is a sequence of matrices such that their product is $M$ and each is either a permutation matrix or a matrix in row- or column-echelon form.

\paragraph*{Flows and cuts.} For a graph $G$ and disjoint subsets of vertices $S,T\subseteq V(G)$, an {\em{$(S,T)$-path}} is a path in $G$ that starts in a vertex of $S$, ends in a vertex of $T$, and whose internal vertices do not belong to $S\cup T$. In this paper, an {\em{$(S,T)$-vertex flow}} is a family of $(S,T)$-paths $\Ff=\{P_1,P_2,\ldots,P_k\}$ that are internally vertex-disjoint; note that we do allow the paths to share endpoints in $S$ or $T$. The size of a flow $\Ff$, denoted $|\Ff|$, is the number of paths in it. A subset $X\subseteq V(G)\setminus (S\cup T)$ is an {\em{$(S,T)$-vertex cut}} if no vertex of $T$ is reachable by a path from some vertex of $S$ in the graph $G-X$. A variant of the well-known Menger's theorem states that the maximum size of an $(S,T)$-vertex flow is always equal to the minimum size of an $(S,T)$-vertex cut, provided there is no edge between $S$ and $T$. In case $G$ is directed, instead of undirected paths, we consider directed paths starting from $S$ and ending in $T$, and the same statement of Menger's theorem holds (the last condition translates to the nonexistence of edges from $S$ to $T$). Note that in this definitions we are only interested in flows and cuts in unweighted graphs, or in other words, we put unit capacities on all the vertices. 

There is a wide variety of algorithms for computing the maximum vertex flows and minimum vertex cuts in undirected/directed graphs in polynomial time. Among them, the most basic is the classic algorithm of Ford and Fulkerson, which uses the technique of finding consecutive augmentations of an $(S,T)$-vertex flow, up to the moment when a maximum flow is found. More precisely, the following well-known result is used.

\begin{theorem}[Max-flow augmentation]\label{thm:augm} 
There exists an algorithm that, given a directed graph $G$ on $n$ vertices and $m$ edges, disjoint subsets $S,T\subseteq V(G)$ with no edge from $S$ to $T$, and some $(S,T)$-vertex flow $\Ff$, works in $\Oh(n+m)$ time and either certifies that $\Ff$ is maximum by providing an $(S,T)$-vertex cut of size $|\Ff|$, or finds an $(S,T)$-vertex flow $\Ff'$ with $|\Ff'|=|\Ff|+1$.
\end{theorem}

The classic proof of Theorem~\ref{thm:augm} works as follows: the algorithm first constructs the {\em{residual network}} that encodes where more flow could be pushed. Then, using a single BFS it looks for an {\em{augmenting path}}. If such an augmenting path can be found, then it can be used to modify the flow so that its size increases by one. On the other hand, the nonexistence of such a path uncovers an $(S,T)$-vertex cut of size $|\Ff|$. Of course, the analogue of Theorem~\ref{thm:augm} for undirected graphs follows by turning an undirected graph into a directed one by replacing every edge $uv$ with arcs $(u,v)$ and $(v,u)$. 

The next well-known corollary follows by applying the algorithm of Theorem~\ref{thm:augm} at most $k+1$~times.

\begin{corollary}\label{cor:flow-to-k}
There exists an algorithm that, given an undirected/directed graph $G$ on $n$ vertices and $m$ edges, disjoint subsets $S,T\subseteq V(G)$ with no edge from $S$ to $T$, and a positive integer $k$, works in time $\Oh(k\cdot (n+m))$ and provides one of the following outcomes:
\begin{enumerate}[(a)]
\item a maximum $(S,T)$-vertex flow of size $\ell$ together with a minimum $(S,T)$-vertex cut size $\ell$, for some $\ell\leq k$; or
\item a correct conclusion that the size of the maximum $(S,T)$-vertex flow (equivalently, of the minimum $(S,T)$-vertex cut) is larger than $k$.
\end{enumerate}
\end{corollary}

\paragraph*{Tree decompositions.} We now recall the main concepts of graph decompositions used in this paper. First, we recall standard tree and path decompositions.

\begin{definition}
A {\em{tree decomposition}} of a graph $G$ is a pair $(\Tt,\{B_x\}_{x\in V(\Tt)})$, where $\Tt$ is a tree and each node $x$ of $\Tt$ is associated with a subset of vertices $B_x\subseteq V(G)$, called the {\em{bag at $x$}}. Moreover, the following condition have to be satisfied:
\begin{itemize}
\item For each edge $uv\in E(G)$, there is some $x\in V(\Tt)$ such that $\{u,v\}\subseteq B_x$.
\item For each vertex $u\in V(G)$, define $\Tt[u]$ to be the subgraph of $\Tt$ induced by nodes whose bags contain $u$. Then $\Tt[u]$ is a non-empty and connected subtree of $\Tt$.
\end{itemize}
The {\em{width}} of $\Tt$ is equal to $\max_{x\in V(\Tt)} |B_x|-1$, and the {\em{treewidth of $G$}}, denoted $\tw(G)$, is the minimum possible width of a tree decomposition of $G$. In case $\Tt$ is a path, we call $\Tt$ also a {\em{path decomposition}} of $G$. The {\em{pathwidth of $G$}}, denoted $\pw(G)$, is the minimum possible width of a tree decomposition of $G$.
\end{definition}

We follow the convention that whenever $(\Tt,\{B_x\}_{x\in V(\Tt)})$ is a tree decomposition of $G$, then elements of $V(\Tt)$ are called {\em{nodes}} whereas elements of $V(G)$ are called {\em{vertices}}. Moreover, we often identify the nodes of $\Tt$ with bags associated with them, and hence we can talk about adjacent bags, etc. Also, we often refer to the tree $\Tt$ only as a tree decomposition, thus making the associated family $\{B_x\}_{x\in V(\Tt)}$ of bags implicit.

Throughout this paper we assume that all tree or path decompositions of width $k$ given on input have $\Oh(|V(G)|)$ nodes. In fact, we will assume that the input decompositions are {\em{clean}}, defined as follows.

\begin{definition}
	A tree decomposition $(\Tt,\{B_x\}_{x\in V(\Tt)})$ of $G$ is called {\em{clean}} if for every $xy\in V(\Tt)$, it holds that $B_x\nsubseteq B_y$ and $B_y\nsubseteq B_x$.
\end{definition}

It is easy to see that a clean tree decomposition of a graph on $n$ vertices has at most $n$ nodes, and that any tree (path) decomposition $\Tt$ of width $k$ can be transformed in time $\Oh(k|V(\Tt)|)$ to a {\em{clean}} tree (path) decomposition of the same width (see e.g.~\cite[Lemma 11.9]{FlumGroheBook}). Thus, the input decomposition can be always made clean in time linear in its size.

Finally, we recall the definition of another width parameter we use, see e.g.~\cite{Wood09}.

\begin{definition}
A {\em{tree-partition decomposition}} of a graph $G$ is a pair $(\Tt,\{B_x\}_{x\in V(\Tt)})$, where $\Tt$ is a tree and each node $x$ of $\Tt$ is associated with a subset of vertices $B_x\subseteq V(G)$, called the {\em{bag at $x$}}. Moreover, the following condition have to be satisfied:
\begin{itemize}
\item Sets $\{B_x\}_{x\in V(\Tt)}$ form a partition of $V(G)$, and in particular are pairwise disjoint.
\item For each edge $uv\in E(G)$, either there is some $x\in V(\Tt)$ such that $\{u,v\}\subseteq B_x$, or there is some $xy\in E(\Tt)$ such that $u\in B_x$ and $v\in B_y$.
\end{itemize}
The {\em{width}} of $\Tt$ is equal to $\max_{x\in V(\Tt)} |B_x|$, and the {\em{tree-partition width of $G$}}, denoted $\tpw(G)$, is the minimum possible width of a tree-partition decomposition of $G$.
\end{definition}

It is easy to see that empty bags in a tree-partition decomposition can be disposed of, implying $|V(\Tt)|\leq |V(G)|$ without loss of generality.
We have $1+\tw(G)\leq 2\cdot \tpw(G)$, but $\tpw(G)$ can be arbitrarily large already for graphs of constant treewidth, in case the maximum degree is unbounded~\cite{Wood09}.
We also need the following well-known fact.

\begin{lemma}[cf. Exercise 7.15 in~\cite{platypus}]\label{lem:tw-edges}\label{cor:tpw-edges}
	A graph on $n$ vertices of treewidth $k$ has at most $kn$ edges.
	Hence a graph on $n$ vertices of tree-partition width $k$ has at most $2kn$ edges.
\end{lemma}

\paragraph*{Measures and balanced separators.} In several parts of the paper, we will be introducing auxiliary weight functions on the vertices of graphs, which we call {\em{measures}}. 

\begin{definition}
Let $G$ be a graph. Any function $\meas\colon V(G)\to \R^+\cup \{0\}$ that is positive on at least one vertex is called a {\em{measure on $V(G)$}}. For a subset $A\subseteq V(G)$, we denote $\meas(A)=\sum_{u\in A} \meas(u)$.
\end{definition}

First, we need the following simple folklore lemma about the existence of balanced vertices of trees.

\begin{lemma}\label{lem:bal-bag}
Let $T$ be a tree on $q$ vertices, with a measure $\meas$ defined on its vertex set. Then a node $x\in V(T)$ such that $\meas(V(C))\leq \meas(V(T))/2$ for every $C\in \cc(T-x)$ can be found in time $\Oh(q)$.
\end{lemma}
\begin{proof}
Consider any edge $yz\in E(T)$, and let the removal of $yz$ split $T$ into subtrees $T_y$ and $T_z$, where $y\in V(T_y)$ and $z\in V(T_z)$. Orient $yz$ from $y$ to $z$ if $\meas(V(T_y))<\meas(V(T_z))$, from $z$ to $y$ if $\meas(V(T_y))>\meas(V(T_z))$, and arbitrarily if $\meas(V(T_y))=\meas(V(T_z))$. The obtained oriented tree has $q$ nodes and $q-1$ directed edges, which means that there is a node $x$ that has indegree $0$. For every neighbor $y$ of $x$ we have that the edge $xy$ was directed towards $x$. This means that $\meas(V(T_x))\geq \meas(V(T_y))$; equivalently $\meas(V(T_y))\leq \meas(V(T))/2$. As $y$ was an arbitrarily chosen neighbor of $x$, it follows that $x$ satisfies the required property.

As for the algorithmic claim, it is easy to implement the procedure orienting the edges in time $\Oh(q)$ by using a recursive depth-first search procedure on $T$ that returns the total weight of nodes in the explored subtree.
Having the orientation computed, suitable $x$ can be retrieved by a simple indegree count in time $\Oh(q)$.
\end{proof}

In graphs of bounded treewidth, Lemma~\ref{lem:bal-bag} can be generalized to find balanced bags instead of balanced vertices. 

\begin{definition}
Let $G$ be a graph, let $\meas$ be a measure on $V(G)$, and let $\alpha\in [0,1]$. A set $X\subseteq V(G)$ is called an {\em{$\alpha$-balanced separator w.r.t. $\meas$}} if for each $C\in \cc(G-X)$, it holds that $\meas(V(C))\leq \alpha\cdot \meas(V(G))$.
\end{definition}

\begin{lemma}[Lemma 7.19 of~\cite{platypus}]\label{lem:balanced}
Let $G$ be a graph with $\tw(G)<k$, and let $\meas$ be a measure on $V(G)$. Then there exists a $\frac{1}{2}$-balanced separator $X$ w.r.t. $\meas$ with $|X|\leq k$.
\end{lemma}

%% file: approx.tex
\newcommand{\chld}{\texttt{chld}}
\newcommand{\Ci}{{\textrm{(i)}}}
\newcommand{\Cii}{{\textrm{(ii)}}}
\newcommand{\Ciii}{{\textrm{(iii)}}}
\newcommand{\all}{{\overline{S}}}
\newcommand{\Rt}{\mathfrak{T}}

In this section we show our approximation algorithm for treewidth, i.e., prove Theorem~\ref{thm:apx}. For the reader's convenience, we restate it here.

\restateapx*

In our proof of Theorem~\ref{thm:apx}, we obtain an upper bound of $1800k^2$ on the width of the computed tree decomposition. We remark that this number can be improved by a more careful analysis of different parameters used throughout the algorithm; however, we refrain from performing a tighter analysis in order to simplify the presentation.

We first prove the backbone technical result, that is, an approximation algorithm for finding balanced separators. This algorithm will be used as a subroutine in every step of the algorithm of Theorem~\ref{thm:apx}. Our approach for approximating balanced separators is based on the work of Feige and Mahdian in~\cite{FeigeM06}.

\begin{lemma}\label{lem:sep-apx}
There exists an algorithm that, given a graph $G$ on $n$ vertices and $m$ edges with a measure $\meas$ on $V(G)$, and a positive integer $k$, works in time $\Oh(k^4\cdot (n+m))$ and returns one of the following outcomes:
\begin{enumerate}[(1)]
\item\label{cn:large} A $\frac{7}{8}$-balanced separator $Y$ w.r.t. $\meas$ with $|Y|\leq 100k^2$;
\item\label{cn:small} A $(1-\frac{1}{100k})$-balanced separator $X$ w.r.t. $\meas$ with $|X|\leq k$;
\item\label{cn:no} A correct conclusion that $\tw(G)\geq k$.
\end{enumerate}
\end{lemma}
\begin{proof}
By rescaling $\meas$ if necessary, we assume that $\meas(V(G))=1$. Throughout the proof we assume that $\tw(G)<k$ and hence, by Lemma~\ref{lem:balanced}, there exists some $\frac{1}{2}$-balanced separator $W$ w.r.t $\meas$, which is of course unknown to the algorithm. We will prove that whenever such a $W$ exists, the algorithm reaches one of the outcomes~\eqref{cn:large} or~\eqref{cn:small}. If none of these outcomes is reached, then no such $W$ exists and, by Lemma~\ref{lem:balanced}, the algorithm can safely report that $\tw(G)\geq k$, i.e., reach outcome~\eqref{cn:no}. We also assume that for every vertex $u\in V(G)$ it holds that $\meas(u)<\frac{1}{100k}$, because otherwise we can immediately provide outcome~\eqref{cn:small} by setting $X=\{u\}$. Note that in particular this implies that $\meas(W)<\frac{1}{100}$.

We first generalize the problem slightly. Suppose that for some $i\leq k$ we are given a set $Y_i$ such that the following invariants are satisfied: (i) $|Y_i|\leq 100ik$ and (ii) $|W\cap Y_i|\geq i$. Then, the claim is as follows:
\begin{claim}\label{cl:general}
Given $Y_i$ satisfying invariants (i) and (ii) for some $i<k$, one can in time $\Oh(k^3\cdot (n+m))$ either arrive at one of the outcomes~\eqref{cn:large} or~\eqref{cn:small}, or find a set $Z$ with $|Z|\leq 100k$ and $Z\cap Y_i=\emptyset$, such that $Z\cap W\neq \emptyset$.
\end{claim}
Before we proceed to the proof of Claim~\ref{cl:general}, we observe how Lemma~\ref{lem:sep-apx} follows from it. We start with $Y_0=\emptyset$, which clearly satisfies the invariants (i) and (ii). Then we iteratively compute $Y_1,Y_2,Y_3,\ldots$ as follows: when computing $Y_{i+1}$, we use the algorithm of Claim~\ref{cl:general} to either provide outcome~\eqref{cn:large} or~\eqref{cn:small}, in which case we terminate the whole computation, or find a suitable set $Z$. Then $Y_{i+1}=Y_i\cup Z$ satisfies the invariants (i) and (ii) for the next iteration, and hence we can proceed. Suppose that this algorithm successfully performed $k$ iterations, i.e., it constructed $Y_k$. Then we have that $|W\cap Y_k|\geq k$, so since $|W|\leq k$, we have $W\subseteq Y_k$. Then $Y_k$ should be a $\frac{1}{2}$-balanced separator w.r.t. $\meas$ and $|Y_k|\leq 100k^2$, so it can be reported as $Y$ in outcome~\eqref{cn:large}. If $Y_k$ is not a $\frac{1}{2}$-balanced separator, then $W$ did not exist in the first place and the algorithm can safely report outcome~\eqref{cn:no}. Since the algorithm of Claim~\ref{cl:general} works in time $\Oh(k^{3}\cdot (n+m))$ and we apply it at most $k$ times, the running time promised in the lemma statement follows.

We now proceed to the proof of Claim~\ref{cl:general}. Let $G'=G-Y_i$. First, let us investigate the connected components of $G'$. If $\meas(V(C))\leq \frac{7}{8}$ for each $C\in \cc(G')$, then we can reach outcome~\eqref{cn:large} by taking $Y=Y_i$, because $|Y_i|\leq 100ik\leq 100k^2$. Hence, suppose there is a connected component $C_0\in \cc(G')$ with $\meas(C_0)>\frac{7}{8}$. Let $T_0$ be an arbitrary spanning tree of $C_0$, and let $V_0=V(C_0)$. We first prove an auxiliary claim that will imply that we can find a nice partitioning of $T_0$. This is almost exactly the notion of {\em{Steiner decompositions}} used by Feige and Mahdian~\cite{FeigeM06} (see Definition~5.1 and Lemma~5.2 in~\cite{FeigeM06}), but we choose to reprove the result for the sake of completeness.
\begin{claim}\label{cl:tree-partition}
Suppose we are given a number $\lambda\in \R^+$ and a tree $T$ on $n$ vertices with a measure $\meas$ on $V(T)$, such that $\meas(u)<\lambda$ for each $u\in V(T)$. Then one can in time $\Oh(n)$ find a family $\Ff=\{(R_1,u_1), (R_2,u_2),\ldots,(R_p,u_p)\}$ (with $u_i$s not necessarily distinct) such that the following holds (in the following, we denote $\tilde{R}_i=R_i\setminus \{u_i\}$):
\begin{enumerate}[(a)]
\item\label{cn:each_part} for each $i=1,2,\ldots,p$, we have that $u_i\in R_i\subseteq V(T)$, $T[R_i]$ is connected, and $\lambda\leq \meas(\tilde{R}_i)<4\lambda$;
\item\label{cn:almost_all} $\meas(V(T)\setminus \bigcup_{i=1}^p \tilde{R}_i)<2\lambda$;
\item\label{cn:disjoint} sets $\tilde{R}_i$ are pairwise disjoint for $i=1,2,\ldots,p$.
\end{enumerate}
\end{claim}
\begin{proof}
We first provide a combinatorial proof of the existence of $\Ff$, which proceeds by induction on $n$. Then we will argue how the proof gives rise to an algorithm constructing $\Ff$ in linear time. 

Let us root $T$ in an arbitrary vertex $r$, which imposes a parent-child relation on the vertices of $T$. For $v\in V(T)$, let $T_v$ be the subtree rooted at $v$. If $\meas(V(T))<2\lambda$ then we can take $\Ff=\emptyset$, so suppose otherwise. Let then $u$ be the deepest vertex of $T$ for which $\meas(V(T_u))\geq 2\lambda$, and let $u_1,u_2,\ldots,u_q$ be the children of $u$. Let us denote $a_j=\meas(V(T_{u_j}))$, for $j=1,2,\ldots,q$. As $u$ was chosen to be the deepest, we have that $a_j<2\lambda$ for each $j=1,2,\ldots,q$. Moreover, since $\meas(u)<2\lambda$ and $\meas(V(T_u))\geq 2\lambda$, we have that $u$ has at least one child, i.e., $q>0$.

Scan the sequence $a_1,a_2,\ldots,a_q$ from left to right, and during this scan iteratively extract minimal prefixes for which the sum of entries is at least $\lambda$, up to the point when the total sum of the remaining numbers is smaller than $\lambda$. Let these extracted sequences be $$\{a_1,a_2,\ldots,a_{j_2-1}\},\{a_{j_2},a_{j_2+1},\ldots,a_{j_3-1}\},\ldots,\{a_{j_s},a_{j_s+1},\ldots,a_{j_{s+1}-1}\},$$ where $s$ is their number. Then, denoting $a_{j_1}=1$, we infer from the construction and the fact that $a_j<2\lambda$ for all $j$ that the following assertions hold:
\begin{itemize}
\item $\lambda\leq \sum_{i=j_\ell}^{j_{\ell+1}-1} a_i<3\lambda$ for all $\ell=1,2,\ldots,s$; and
\item $\sum_{i=j_{s+1}}^q a_i<\lambda$.
\end{itemize}
Moreover, since $\meas(V(T_u))\geq 2\lambda$, we have that $\sum_{i=1}^q a_i\geq 2\lambda-\meas(u)>\lambda$, and hence at least one sequence has been extracted; i.e., $s\geq 1$.

For $\ell=1,\ldots,s-1$, let $I_\ell=\{j_\ell,j_\ell+1,\ldots,j_{\ell+1}-1\}$, and let $I_s=\{j_s,j_s+1,\ldots,q\}$. Concluding, from the assertions above it follows that we have partitioned $[q]$ into contiguous sets of indices $I_1,I_2,\ldots,I_s$ such that $s\geq 1$ and $\lambda\leq \sum_{i\in I_\ell} a_i < 4\lambda$, for each $\ell=1,2,\ldots,s$.

Let $T'$ be $T$ with all the subtrees $T_{u_j}$ removed, for $j=1,2,\ldots,q$. Since $q>0$, we have that $|V(T')|<|V(T)|$, and hence by the induction hypothesis we can find a family $\Ff'=\{(R_1,u_1), (R_2,u_2),\ldots,(R_{p'},u_{p'})\}$ for the tree $T'$ satisfying conditions~\eqref{cn:each_part},~\eqref{cn:almost_all} and~\eqref{cn:disjoint}.

For $\ell=p'+1,p'+2,\ldots,p'+s$, let $R_{p'+\ell}=\{u\}\cup \bigcup_{i\in I_\ell} V(T_{u_i})$ and $u_\ell=u$. Observe that $\meas(R_{p'+\ell}\setminus \{u_{p'+\ell}\})=\sum_{i\in I_\ell} a_i$, so we obtain that $\lambda\leq \meas(\tilde{R}_{p'+\ell})<4\lambda$. Consider $\Ff=\Ff'\cup \{(R_{p'+1},u_{p'+1}),(R_{p'+2},u_{p'+2}),\ldots,(R_{p'+s},u_{p'+s})\}$. It can be easily verified that conditions~\eqref{cn:each_part},~\eqref{cn:almost_all}, and~\eqref{cn:disjoint} hold for $\Ff$ using the induction assumption that they held for $\Ff'$. This concludes the inductive proof of the existence of $\Ff$.

Naively, the proof presented above gives rise to an $\Oh(n^2)$ algorithm that iteratively finds a suitable vertex $u$, and applies itself recursively to $T'$. However, it is very easy to see that the algorithm can be implemented in time $\Oh(n)$ by processing $T$ bottom-up, remembering the total measure of the vertices in the processed subtree, and cutting new pairs $(R_i,u_i)$ whenever the accumulated measure exceeds $2\lambda$.
\cqed\end{proof}

Armed with Claim~\ref{cl:tree-partition}, we proceed with the proof of Claim~\ref{cl:general}. Apply the algorithm of Claim~\ref{cl:tree-partition} to tree $T_0$ and $\lambda = \frac{1}{100k}$. Note that the premise of the claim holds by the assumption that $\meas(u)<\frac{1}{100k}$ for each $u\in V(G)$. Therefore, we obtain a family $\Ff=\{(R_1,u_1), (R_2,u_2),\ldots,(R_p,u_p)\}$ satisfying conditions~\eqref{cn:each_part},~\eqref{cn:almost_all}, and~\eqref{cn:disjoint} for $T_0$. For $i=1,2,\ldots,p$, let $\tilde{R}_i=R_i\setminus \{u_i\}$ and let $Z=\{u_1,u_2,\ldots,u_p\}$. Note that by conditions~\eqref{cn:each_part} and~\eqref{cn:disjoint} we have that
$$1=\meas(V(G))\geq \sum_{i=1}^p \meas(\tilde{R}_i)\geq p\lambda=\frac{p}{100k},$$
so $|Z|\leq p\leq 100k$. 

We need the following known fact.

\begin{claim}[cf. the proof of Lemma~7.20 in \cite{platypus}]\label{cl:grouping}
Let $b_1,b_2,\ldots,b_q$ be nonnegative reals such that $\sum_{i=1}^q b_i\leq 1$ and $b_i\leq 1/2$ for each $i=1,2,\ldots,q$. Then $[q]$ can be partitioned into two sets $J_1$ and $J_2$ such that $\sum_{i\in J_z} b_i\leq \frac{2}{3}$ for each $z\in \{1,2\}$.
\end{claim}

Let $B_1,B_2,\ldots,B_q$ be the connected components of $C_0-W$, and let $b_i=\meas(V(B_i))$ for each $i=1,2,\ldots,q$. Clearly $\sum_{i=1}^q b_i\leq \meas(V(G))=1$. Moreover, each component $B_i$ is contained in some component of $G-W$, and hence, since $W$ is a $\frac{1}{2}$-balanced separator of $G$ w.r.t. measure $\meas$, we have that $b_i\leq 1/2$ for each $i=1,2,\ldots,q$. By Claim~\ref{cl:grouping} we can find a partition $(J_1,J_2)$ of $[q]$ such that $\sum_{i\in J_z} b_i\leq \frac{2}{3}$ for each $z\in \{1,2\}$. Let $A_1=\bigcup_{i\in J_1} V(B_i)$ and $A_2=\bigcup_{i\in J_2} V(B_i)$. Then vertex sets $A_1$ and $A_2$ are non-adjacent in $G$ and $\meas(A_1),\meas(A_2)\leq \frac{2}{3}$. Recall that $\frac{7}{8}<\meas(V(C_0))=\meas(A_1)+\meas(A_2)+\meas(W\cap V(C_0))$. Since $\meas(W)<\frac{1}{100}$, we have that $\meas(A_1)+\meas(A_2)>\frac{7}{8}-\frac{1}{100}>\frac{5}{6}$. Hence it follows that $\meas(A_1),\meas(A_2)>\frac{5}{6}-\frac{2}{3}=\frac{1}{6}$.

Let $K_1\subseteq [q]$ be the set of those indices $i\in [q]$ for which $\tilde{R}_i\cap A_1\neq\emptyset$. By properties~\eqref{cn:each_part},~\eqref{cn:almost_all}, and~\eqref{cn:disjoint}, we obtain that:
$$\frac{1}{6}<\meas(A_1)\leq \meas\left(V(C_0)\setminus \bigcup_{i=1}^p \tilde{R}_i\right)+\sum_{i\in K_1} \meas(\tilde{R}_i)<2\lambda+4\lambda|K_1|.$$
Since $\lambda=\frac{1}{100k}$, we have that
$$|K_1|>\frac{1}{24\lambda}-\frac{1}{2}>3k.$$
Since $|W|\leq k$ and sets $\tilde{R}_i$ are pairwise disjoint, there is $L_1\subseteq K_1$ of size at least $2k$ such that additionally $\tilde{R}_i\cap W=\emptyset$ for each $i\in L_1$. Symmetrically we prove that there is a set $L_2\subseteq [q]$ of at least $2k$ indices such that $\tilde{R}_i\cap A_2\neq\emptyset$ and $\tilde{R}_i\cap W=\emptyset$, for each $i\in L_2$.

Suppose for a moment that $Z\cap W=\emptyset$. Then, for each $i\in L_1$ we in fact have that $R_i\cap W=\emptyset$. Since $G[R_i]$ is connected, $\tilde{R}_i\cap A_1\neq\emptyset$, and there is no edge between $A_1$ and $A_2$, it follows that $R_i\subseteq A_1$. Similarly, $R_i\subseteq A_2$ for each $i\in L_2$.

Therefore, the algorithm does as follows. For each pair of distinct indices $i,j\in [q]$ for which $R_i\cap R_j=\emptyset$, we verify whether the size of a minimum vertex cut in $G$ between $R_i$ and $R_j$ does not exceed $k$. If we find such a pair and the corresponding vertex cut $X$, then $X$ separates $R_i$ from $R_j$, so $X$ is a $(1-\frac{1}{100k})$-balanced separator w.r.t. $\meas$ in $G$, due to $\meas(R_i),\meas(R_j)\geq \frac{1}{100k}$. Therefore, such $X$ can be reported as outcome~\eqref{cn:small}. The argumentation of the previous paragraph ensures us that at least one such pair $(i,j)$ will be found provided $Z\cap W=\emptyset$.

Hence, if for every such pair $(i,j)$ the minimum vertex cut in $G$ between $R_i$ and $R_j$ is larger than $k$, then we have a guarantee that $Z\cap W\neq \emptyset$. Since $|Z|\leq 100k$, we can provide the set $Z$ as the outcome of the algorithm of Claim~\ref{cl:general}. This concludes the description of the algorithm of Claim~\ref{cl:general}
. To bound its running time, observe that the application of the algorithm of Claim~\ref{cl:tree-partition} takes time $\Oh(n+m)$, whereas later we verify the minimum value of a vertex cut for at most $p^2=\Oh(k^2)$ pairs $(i,j)$. By Corollary~\ref{cor:flow-to-k}, each such verification can be implemented in time $\Oh(k\cdot (n+m))$. Hence, the whole algorithm of Claim~\ref{cl:general} indeed runs in time $\Oh(k^3\cdot (n+m))$. As argued before, Lemma~\ref{lem:sep-apx} follows from Claim~\ref{cl:general}.
\end{proof}

Having an approximation algorithm for balanced separators, we can proceed with the proof of Theorem~\ref{thm:apx}. Following the approach introduced by Robertson and Seymour~\cite{gm13}, we solve a more general problem. Let $\eta=100k^2$. Assume we are given a graph $H$ together with a subset $S\subseteq V(H)$, $S\neq V(H)$, with the invariant that $|S|\leq 17\eta$. The goal is to either conclude that $\tw(H)\geq k$, or to compute a rooted tree decomposition of $H$, i.e., a tree decomposition rooted at some node $r$, that has width at most $18\eta$ and where $S$ is a subset of the root bag. Note that if we find an algorithm with running time $\Oh(k^7\cdot n\log n)$ for the generalized problem, then Theorem~\ref{thm:apx} will follow by applying it to $G$ and $S=\emptyset$. 

The algorithm for the more general problem works as follows. First, observe that we can assume that $|E(H)|\leq kn$, where we denote $n=|V(H)|$. Indeed, by Lemma~\ref{lem:tw-edges} we can immediately answer that $\tw(H)\geq k$ if this assertion does not hold. Having this assumption, we consider three cases: either (i) $|V(H)\setminus S|\leq \eta$, or, if this case does not hold, (ii) $|S|\leq 16\eta$, or (iii) $16\eta<|S|\leq 17\eta$.
\bigskip

\noindent{\bf{Case (i)}}: If $|V(H)\setminus S|\leq \eta$, we can output a trivial tree decomposition with one bag containing $V(H)$, because the facts that $|V(H)\setminus S|\leq \eta$ and $|S|\leq 17\eta$ imply that $|V(H)|\leq 18\eta$. In the other cases we assume that (i) does not hold, i.e., $|V(H)\setminus S|>\eta$.

\medskip

\noindent{\bf{Case (ii)}}: Suppose $|S|\leq 16\eta$. Define a measure $\meas_{\all}$ on $V(H)$ as follows: $\meas_{\all}(u)=0$ for each $u\in S$, and $\meas_{\all}(u)=1$ for each $u\notin S$. Run the algorithm of Lemma~\ref{lem:sep-apx} on $H$ with measure $\meas_{\all}$. If it concluded that $\tw(H)\geq k$, then we can terminate and pass this answer. Otherwise, let $X$ be the obtained subset of vertices. Regardless whether outcome~\eqref{cn:large} or~\eqref{cn:small} was given, we have that $X$ is a $(1-\frac{1}{100k})$-balanced separator in $H$ w.r.t. measure $\meas_{\all}$, and moreover $|X|\leq \eta$. Let $B=S\cup X$ and observe that $|B|\leq |S|+|X|\leq 17\eta$. Consider the connected components of $H-B$. For each $C\in \cc(H-B)$, define a new instance $(H_C,S_C)$ of the generalized problem as follows: $H_C=H[N_H[V(C)]]$ and $S_C=N_H(V(C))$. Observe that $S_C\subseteq B$, hence $|S_C|\leq |B|\leq 17\eta$, and thus the invariant that $|S_C|\leq 17\eta$ is satisfied in the instance $(H_C,S_C)$. Apply the algorithm recursively to $(H_C,S_C)$, yielding either a conclusion that $\tw(H_C)\geq k$, in which case we can report that $\tw(H)\geq k$ as well, or a rooted tree decomposition $\Tt_C$ of $H_C$ that has width at most $18\eta$ and where $S_C$ is a subset of the root bag. Construct a tree decomposition $\Tt$ of $H$ as follows: create a root node $r$ associated with bag $B$, and, for each $C\in \cc(H-B)$, attach the decomposition $\Tt_C$ below $r$ by making the root of $\Tt_C$ a child of $r$. It easy to verify that $\Tt$ created in this manner indeed is a tree decomposition of $H$, and its width does not exceed $18\eta$ because $|B|\leq 17\eta\leq 18\eta$.

\medskip

\noindent{\bf{Case (iii)}}: Suppose $16\eta<|S|\leq 17\eta$. Define a measure $\meas_{S}$ on $V(H)$ as follows: $\meas_{S}(u)=1$ for each $u\in S$, and $\meas_S(u)=0$ for each $u\notin S$. Run the algorithm of Lemma~\ref{lem:sep-apx} on $H$ with measure $\meas_{S}$. Again, if it concluded that $\tw(H)\geq k$, then we can terminate and pass this answer. Otherwise, we obtain either a $\frac{7}{8}$-balanced separator $Y$ with $|Y|\leq \eta$, or a $(1-\frac{1}{100k})$-balanced separator $X$ with $|X|\leq k$. Let $Z$ be this separator, being either $X$ or $Y$ depending on the subcase. The algorithm proceeds as in the previous case. Define $B=S\cup Z$; then $|B|\leq |S|+|Z|\leq 18\eta$. For each $C\in \cc(H-B)$, define a new instance $(H_C,S_C)$ of the generalized problem by taking $(H_C,S_C)=(H[N_H[V(C)]],N_H(V(C)))$. Apply the algorithm recursively to each $(H_C,S_C)$, yielding either a conclusion that $\tw(H_C)\geq k$, which implies $\tw(H)\geq k$, or a tree decomposition $\Tt_C$ of $H_C$ that has width at most $18\eta$ and $S_C$ is contained in its root bag. Construct the output decomposition $\Tt$ by taking $B$ as the bag of the root node $r$, and attaching all the decompositions $\Tt_C$ below $r$ by making their roots children of $r$. Again, it can be easily verified that $\Tt$ is a tree decomposition of $H$ of width at most $18\eta$. The only verification that was not performed is that in the new instances $(H_C,S_C)$, the invariant that $|S_C|\leq 17\eta$ still holds. We shall prove an even stronger fact: that $|S_C|\leq |S|-k$, for each $C\in \cc(H-B)$, so the needed invariant will follow by $|S|\leq 17\eta$. However, we will need this stronger property in the future. The proof investigates the subcases $Z=X$ and $Z=Y$ separately.

Suppose first that $Z=X$, that is, the algorithm of Lemma~\ref{lem:sep-apx} have found a $(1-\frac{1}{100k})$-balanced separator $X$ with $|X|\leq k$. Consider any connected component $C\in \cc(H-B)$, and let $D$ be the connected component of $H-X$ in which $C$ is contained. Since $X$ is $(1-\frac{1}{100k})$-balanced w.r.t. $\meas_S$, we have that 
$$|S\cap D|=\meas_S(D)\leq \left(1-\frac{1}{100k}\right)\cdot\meas_S(V(H))=\left(1-\frac{1}{100k}\right)\cdot |S|=|S|-\frac{|S|}{100k}<|S|-16k;$$
the last inequality follows from the assumption that $|S|>16\eta=1600k^2$. Now observe that $N_H(C)\subseteq X\cup (S\cap D)$, so $|N_H(C)|<k+(|S|-16k)=|S|-15k<|S|-k$.

Suppose second that $Z=Y$, that is, the algorithm of Lemma~\ref{lem:sep-apx} found a $\frac{7}{8}$-balanced separator $Y$ with $|Y|\leq \eta$. Again, consider any connected component $C\in \cc(H-B)$, and let $D$ be the connected component of $H-Y$ in which $C$ is contained. Since $Y$ is $\frac{7}{8}$-balanced w.r.t. $\meas_S$, we have that
$$|S\cap D|=\meas_S(D)\leq \frac{7}{8}\cdot \meas_S(V(H))=\frac{7}{8}|S|=|S|-\frac{|S|}{8}<|S|-2\eta;$$
the last inequality follows from the assumption that $|S|>16\eta$. Again, observe that $N_H(C)\subseteq Y\cup (S\cap D)$, so $|N_H(C)|<\eta+(|S|-2\eta)=|S|-\eta<|S|-k$.
\bigskip

This concludes the description of the algorithm. Its partial correctness is clear: if the algorithm concludes that $\tw(H)\geq k$, then it is always correct, and similarly any tree decomposition of $H$ output by the algorithm satisfies the specification. We are left with arguing that the algorithm always stops (i.e., it does not loop in the recursion) and its running time is bounded by $\Oh(k^7\cdot n\log n)$. We argue both these properties simultaneously.

Let $\Rt$ be the recursion tree yielded by the algorithm. That is, $\Rt$ is a rooted tree with nodes corresponding to recursive subcalls to the algorithm (in case the algorithm loops, $\Rt$ would be infinite). Thus, each node is labeled by the instance $(H',S')$ which is being solved in the subcall. The root of $\Rt$ corresponds to the original instance $(H,S)$, and the children of each node $x$ correspond to subcalls invoked in the call corresponding to $x$. Observe that if the algorithm returns some decomposition $\Tt$ of $H$, then $\Tt$ is isomorphic to $\Rt$, because every subcall produces exactly one new bag, and the bags are arranged into $\Tt$ exactly according to the recursion tree of the algorithm. The leaves of $\Rt$ correspond to calls where no subcall was invoked: either ones conforming to Case~(i), or ones conforming to Case~(ii) or~(iii) when $B=V(H)$. Partition the node set of $\Rt$ into sets $A_{\Ci}$, $A_{\Cii}$, $A_{\Ciii}$, depending whether the corresponding subcall falls into Case~(i),~(ii), or~(iii). For each node $x\in V(\Rt)$, let $(H_x,S_x)$ be its corresponding subcall, let $h_x=|V(H_x)\setminus S_x|$, $s_x=|S_x|$, and $n_x=h_x+s_x=|V(H_x)|$. By $\chld(x)$ we denote the set of children of $x$. The following claim shows the crucial properties of $\Rt$ that follow from the algorithm.

\begin{claim}\label{cl:recursion-bound}
The following holds:
\begin{enumerate}[(a)]
\item\label{as:ii} For each $x\in A_\Cii$ and each $y\in \chld(x)$, we have $h_y\leq (1-\frac{1}{100k})\cdot h_x$.
\item\label{as:iii} For each $x\in A_\Ciii$ and each $y\in \chld(x)$, we have $s_y\leq s_x-k$.
\end{enumerate}
\end{claim}
\begin{proof}
Assertion~\eqref{as:iii} is exactly the stronger property $|N_H(C)|\leq |S|-k$ that we have ensured in Case (iii). Hence, it remains to prove assertion~\eqref{as:ii}. Let $X\subseteq V(H_x)$ be the separator found by the algorithm when investigating the subcall in node $x$. Suppose that child $y$ corresponds to some subcall $(H_y,S_y)=(H_C,S_C)=(N_{H_x}[V(C)],N_{H_x}(V(C)))$, for some component $C\in \cc(H_x-(S_x\cup X))$. Let $D$ be the connected component of $H-X$ in which $C$ is contained. Since $X$ is $(1-\frac{1}{100k})$-balanced in $H_x$ w.r.t. measure $\meas_\all$, we have that 
$$h_y=|C|\leq |D\setminus S|=\meas_\all(D)\leq \left(1-\frac{1}{100k}\right)\cdot\meas_\all(V(H_x))=\left(1-\frac{1}{100k}\right)\cdot h_x.$$\cqed\end{proof}

Using Claim~\ref{cl:recursion-bound}, we can show an upper bound on the depth of $\Rt$. This in particular proves that the algorithm always stops (equivalently, $\Rt$ is finite).

\begin{claim}\label{cl:depth}
The depth of $\Rt$ is bounded by $\Oh(k^2\cdot \log n)$.
\end{claim}
\begin{proof}
Take any finite path $P$ in $\Rt$ that starts in its root and travels down the tree. Since each node of $A_{\Ci}$ is a leaf in $\Rt$, at most one node of $P$ can belong to $A_{\Ci}$. We now estimate how many nodes of $A_{\Cii}$ and $A_{\Ciii}$ can appear on $P$.

First, we claim that there are at most $\Oh(k\cdot \log n)$ nodes of $A_\Cii$ on $P$. Let $x_0,x_1,x_2,\ldots,x_p$ be consecutive vertices of $A_\Cii$ on $P$, ordered from the root. By Claim~\ref{cl:recursion-bound}\eqref{as:ii} we have that $h_{x_{i+1}}\leq (1-\frac{1}{100k})\cdot h_{x_i}$ for each nonnegative integer $i$. Since $h_{x_0}\leq n$, by a trivial induction it follows that $h_{x_i}\leq (1-\frac{1}{100k})^i\cdot n$. However, $h_x\geq 1$ for each $x\in V(\Rt)$, so $p\leq \log_{1-\frac{1}{100k}} 1/n\leq \Oh(k\cdot\log n)$. Consequently, $|V(P)\cap A_\Cii|=p+1\leq \Oh(k\cdot \log n)$.

Second, we claim that on $P$ there cannot be more than $100k$ consecutive vertices from $A_\Ciii$. Assume otherwise, that there exists a sequence $x_0,x_1,\ldots,x_{100k}$ where $x_i\in V(P)\cap \Ciii$ for each $0\leq i\leq 100k$ and $x_{i+1}\in\chld(x_i)$ for each $0\leq i<100k$. By Claim~\ref{cl:recursion-bound}\eqref{as:iii} we have that $s_{x_{i+1}}\leq s_{x_i}-k$, for each $0\leq i<100k$. By a trivial induction we obtain that $s_{x_i}\leq s_{x_0}-ik$. However, $s_{x_0}\leq 17\eta$, so $s_{x_{100k}}\leq 17\eta-100k\cdot k=16\eta$. This is a contradiction with $x_{100k}\in A_\Ciii$, because Case~(ii) should apply to instance $(H_{x_{100k}},S_{x_{100k}})$ instead.

Combining the observations of the previous paragraphs, we see that $P$ has at most $\Oh(k\cdot\log n)$ vertices of $A_\Cii$, the segments of consecutive vertices of $A_\Ciii$ have length at most $100k$, and only the last vertex can belong to $A_\Ci$. It follows that $|V(P)|\leq \Oh(k^2\cdot \log n)$. Since $P$ was chosen arbitrarily, we infer that the depth of $\Rt$ is bounded by $\Oh(k^2\cdot \log n)$.
\cqed\end{proof}

By Claim~\ref{cl:depth}, we already know that the algorithm stops and outputs some tree decomposition~$\Tt$. As we noted before, $\Tt$ and $\Rt$ are isomorphic by a mapping that associates a bag of $\Tt$ with the subcall in which it was constructed. By somewhat abusing the notation, we will identify $\Tt$ with $\Rt$ in the sequel. It remains to argue that the running time of the algorithm is bounded by $\Oh(k^7\cdot n\log n)$.

We partition the total work used by the algorithm between the nodes of $\Tt$. For $x\in V(\Tt)$, we assign to $x$ the operations performed by the algorithm when constructing the bag $B_x$ associated with~$x$: running the algorithm of Lemma~\ref{lem:sep-apx} to compute an appropriate balanced separator, construction of the subcalls, and gluing the subdecompositions obtained from the subcalls below the constructed bag $B_x$. From the description of the algorithm and Lemma~\ref{lem:sep-apx} it readily follows that the work associated with node $x\in V(\Tt)$ is bounded by $\Oh(n_x)$ whenever $x$ belongs $A_\Ci$, and by $\Oh(k^5\cdot n_x)$ whenever $x$ belongs to $A_\Cii$ or $A_\Ciii$. Note that here we use the assumption that $|E(H_x)|\leq k\cdot n_x$ for each node $x\in V(\Tt)$; we could assume this because otherwise the application of the algorithm to instance $(H_x,S_x)$ would immediately reveal that $\tw(H_x)\geq k$.

Take any node $x\in A_\Ci$, and observe that, since $S_x\neq V(H_x)$, we have that $B_x=V(H_x)$ contains some vertex of $H$ which does not belong to any other bag of $\Tt$: any vertex of $V(H_x)\setminus S_x$ has this property. Therefore, the total number of nodes in $A_\Ci$ is at most $n$. Since $n_x\leq \Oh(k^2)$ for each such~$x$, we infer that the total work associated with bags of $A_\Ci$ is bounded by $\Oh(k^2\cdot n)$.

We now bound the work associated with the nodes of $A_\Cii\cup A_\Ciii$. Let $d$ be the depth of $\Tt$; by Claim~\ref{cl:depth} we know that $d\leq \Oh(k^2\cdot \log n)$. For $0\leq i\leq d$, let $L_i$ be the set of nodes of $A_\Cii\cup A_\Ciii$ that are at depth exactly $i$ in $\Tt$. Fix some $i$ with $0\leq i\leq d$. Observe that among $x\in L_i$, the sets $V(H_x)\setminus S_x$ are pairwise disjoint. Hence, $\sum_{x\in L_i} h_x\leq n$. However, for every node $x\in A_\Cii\cup A_\Ciii$ we have that $h_x>\eta$ (or otherwise case~(i) would apply in the subcall corresponding to $x$) and $s_x\leq 17\eta$. This implies that $n_x=h_x+s_x\leq 18h_x$, and hence $\sum_{x\in L_i} n_x\leq 18n$. Consequently, the total work associated with the nodes of $L_i$ is bounded by $\sum_{x\in L_i} \Oh(k^5\cdot n_x)\leq \Oh(k^5\cdot n)$. Since $d=\Oh(k^2\log n)$, we conclude that the total work associated with bags of $A_\Cii\cup A_\Ciii$ is $\Oh(k^7\cdot n\log n)$.

Concluding, we have shown that the algorithm for the generalized problem is correct, always terminates, and in total uses time $\Oh(k^7\cdot n\log n)$. Theorem~\ref{thm:apx} follows by applying it to $H=G$ and $S=\emptyset$.

%% file: gaussian.tex
\def\Cols{\ensuremath{\mathcal{I}}\xspace}
\def\Rows{\ensuremath{\mathcal{J}}\xspace}

First, in Section~\ref{sec:elim} we describe our Gaussian elimination algorithm guided by an ordering of rows and columns of the matrix. Next, in Section~\ref{sec:pw-elim} we show how suitable orderings of low width can be recovered from low-width path and tree-partition decompositions of the bipartite graph associated with the matrix. Finally, in Section~\ref{sec:tw-elim} we present how the approach can be lifted to tree decompositions using an adaptation of the vertex-splitting/matrix sparsification technique.

\subsection{Gaussian elimination with strong orderings}\label{sec:elim}

\paragraph*{Algebraic description.} We first describe our algorithm purely algebraically, showing what arithmetic operations will be performed and in which order. We will address implementation details later.

We assume that the algorithm is given an $n\times m$ matrix $M$ over some field $\F$, and moreover there is an ordering $\mleq$ imposed on the rows and columns of $M$. We will never compare rows with columns using $\mleq$, so actually we are only interested in the orders imposed by $\mleq$ on the rows and on the columns of $M$. The following Algorithm~\ref{algo:gaussian} presents our procedure.

\noindent\begin{algorithm}[h!]
  \KwIn{An $n\times m$ matrix $M$ with an order $\mleq$ on rows and columns of $M$} 
  \KwOut{Matrices $U$, $L$, and removal orders of rows and columns of $M$}  \Indp \BlankLine
  Set \Cols, \Rows to be the sets of all columns and all rows of $M$, respectively.\\
  Set $L$ to be the $n\times n$ identity matrix.\\
  \While{\Cols is not empty}{
        Let $i$ be the earliest column in \Cols (in the ordering $\mleq$).\\
        Let $j$ be the earliest row in \Rows such that $M[j,i]\neq 0$  (in the ordering $\mleq$);\\
		if there is no such row, remove $i$ from \Cols and choose $i$ again.\\
	\For{every row $k\neq j$ in \Rows such that $M[k,i]\neq 0$}{
		/* \textsf{\footnotesize{ Eliminate entry $M[k,i]$ using row $j$. } }*/\\
		Set $L[k,j] := \frac{M[k,i]}{M[j,i]}$.\\
		\For{every column $\ell$ such that $M[j,\ell]\neq 0$}{
			Set $M[k,\ell] := M[k,\ell] - L[k,j] \cdot  M[j,\ell]$. \hspace{3.5cm}/* \textsf{\footnotesize{ Entry Update } }*/
		}
		Remove $j$ from \Rows and $i$ from \Cols. \hspace{2.2cm}/* \textsf{\footnotesize{After this, $M[k,i]=0$ for all $k \in \Rows$. } }*/
	}
  }
  Remove all remaining rows from \Rows. \hspace{3.2cm}/*\textsf{\footnotesize{ Note these must be empty rows. } }*/\\
  Set $U:=M$.\\
  \textbf{Return} $U$, $L$, and orders in which the columns/rows were removed from $\Cols$/$\Rows$, respectively.
\vskip 0.5cm
\caption{Gaussian elimination of an $n\times m$ matrix $M$ with order $\mleq$ (see Figure~\ref{fig:elimination}).}
  \label{algo:gaussian}
\end{algorithm}

\begin{figure}[H]
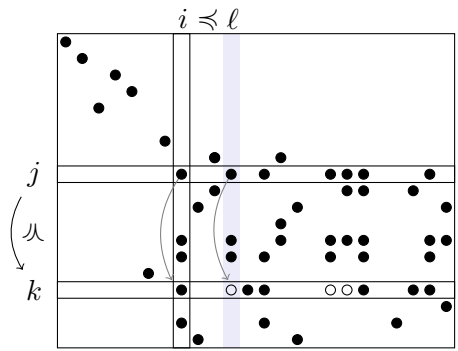

	\centering
	\vspace{-\baselineskip}
	\include{figureElimination}
	\vspace{-\baselineskip}
	\caption{The positions of non-zero entries (black circles) after a few steps of the algorithm. In the current step the entry in column $i$, row $k$ will be eliminated using the earlier row $j$ -- entries that will become non-zero are marked as white circles. }
	\label{fig:elimination}
\end{figure}

In Algorithm~\ref{algo:gaussian}, we consider consecutive columns of $M$ (in the order of $\mleq$), and for each column $i$ we find the first row $j$ (in the order of $\mleq$) that has a non-zero element in column $i$. Row $j$ is then used to eliminate all the other non-zero entries in column $i$. Hence, whenever some column $i$ is removed, it has a zero intersection with all the rows apart from the said row $j$. Inductively, it follows that in the Entry Update step, we have that $j \in \Rows$ with $M[j,\ell]\neq 0$ imply that $\ell \in \Cols$ and columns are never changed after removal from $\Cols$. Clearly, rows are never changed after removal.

The algorithm also returns the {\em{removal orders}} of columns and rows of $M$, that is, the orders in which they were removed from $\Cols$ and $\Rows$, respectively. For $\Cols$ this order coincides with the order $\mleq$, but for $\Rows$ it depends on the positioning of non-zero entries in the matrix, and can differ from $\mleq$.

We now verify that the output of the algorithm forms a suitable factorization of $M$ and can be used for solving linear equations. In the following, we assume that $L$ and $U$ are given as sequences of non-zero entries in consecutive rows (in any order), so their description takes $\Oh(N)$ indices and field values, where $N$ is the number of non-zero entries in $L$ and $U$.

\begin{lemma}\label{lem:usingGaussian}
	Let $M$ be an $n \times m$ matrix over a field $\F$ and $\mleq$ be any ordering of its rows and columns. 
	Within the same time bounds as needed for executing Algorithm~\ref{algo:gaussian}, we can compute
	a $PLUQ$-factorization of $M$,
	and hence the rank, determinant (if $n=m$), and a maximal nonsingular submatrix of $M$.
	Furthermore, for any $r\in \F^n$, the system of linear equations $Mx = r$ can be solved in $\Oh(N)$ additional time and field operations (either outputting an example solution $x\in \F^m$ or concluding that none exists), where $N$ is the number of non-zero entries in $L$ and $U$.
\end{lemma}
\begin{proof}
	Let $U,L$ be the matrices output by the algorithm.
	The sequence of row operations performed in the algorithm implies that each row $U[k,\cdot]$ of the final matrix is obtained from the row $M[k,\cdot]$ of the original matrix by adding row $U[j,\cdot]$ with multiplier $-L[k,j]$, for each $j$ removed earlier (note that rows $j$ removed later are never added and have $L[k,j]=0$).
	That is, for each $k$ we have
	$U[k,\cdot] = M[k,\cdot] - \sum_{j\neq k} L[k,j] \cdot U[j,\cdot]$, from which it immediately follows that $M = LU$ (recall that $L$ was initialized as an identity matrix).
	
	Let $P,Q$ be $n\times n$ and $m\times m$ permutation matrices, respectively, such that $U'=PUQ$ is the matrix $U$ with rows and columns reordered in the removal order output by the algorithm.
	Then $L' = PLP^{-1}$ is the matrix $L$ reordered so that its rows and its columns both correspond to the rows of $M$ in removal order.
	We claim that $U',L'^{T}$ are in row-echelon form, and hence $M=P^{-1} L' U' Q^{-1}$ gives the desired factorization.
	
	For matrix $L'$, this follows directly from the fact that $L$ has ones on the diagonal, which are preserved by the permutation $P$ of both rows and columns, and $L'$ is lower-triangular, because $L[k,j]$ can be non-zero only if row $k$ is removed after row $j$ of $M$ in the algorithm. Hence $L'^T$ is in row-echelon form.
	
	For matrix $U'$, observe
	that, when a row $j$ is removed in the algorithm, it is removed together with a column $i$ that has a non-zero intersection with $j$, and all rows below, that is, all rows that remain to be removed and are currently in $\Rows$, have a zero intersection with $i$.
	Hence also any column $i'$ to the left of $i$ has a zero intersection with row $j$, since  $j$ was removed strictly after removing $i'$.
	This means that $U'[j,i]$ is the first non-zero element of row $j$ and all elements below it are zero.
	Hence $U'$ is in row-echelon form.
	
	It follows that since $L'$ has only ones on the diagonal and since $U'$ is in row-echelon form, the determinant of $M$ (in case $n=m$) is equal to the product of diagonal values of $U'$ multiplied by the signs of the permutations $P$ and $Q$.
	Similarly, $P,Q,L$ are of full rank, so the rank of $M$ is the number of non-zero rows in $U'$.
	A maximum nonsingular submatrix of $U'$ is given by the non-empty rows and the columns containing the left-most non-zero element of each such row -- the same submatrix is maximal nonsingular for $L'U'$. As $L'$ is lower-triangular with ones on the diagonal,the corresponding positions in $M=P^{-1} L' U' Q^{-1}$ give a maximal nonsingular submatrix of $M$.
	Clearly all the above can be computed within the same time bounds.
	
	Given $r$, to solve $Mx=P^{-1} L' U' Q'^{-1} x= r$ it is enough to solve $Q^{-1} y = r$ trivially to get $y$, then $U'z =y$ to get $z$ and similarly with $L'$ and $P^{-1}$ to get $x$.
	The system $U' z = y$ (for $U'$ in row-echelon form, analogously for $L'$) can be solved easily solved by back-substitution.
	That is, start from $z_i=0$ for every column $i$ of $U'$.
	If for any empty row $j$ of $U'$, $y_j$ is non-zero, output that the system has no solutions.
	For every non-empty row $j$ of $U'$ from the lowest to the highest (so in order opposite to the removal order, with shortest rows first), let $U'[j,i]$ be the first non-zero entry of the row $j$ and set $z_i := \frac{y_j - \sum_{i\mleq \ell} U'[j,\ell] \cdot z_\ell}{U'[j,i]}$.
	A standard check shows that $U' z = y$.
	Exactly one multiplication or division and one addition or subtraction is done for every non-zero element of $U'$ and $L'$.
\end{proof}

\paragraph*{Using strong orderings of small width.} We now introduce width parameters for the ordering $\mleq$, and show how to bound the number of fill-in entries using these parameters.

\begin{definition}
	Let $H$ be a graph. An $H$-ordering is an ordering of $V(H)$. An $H$-ordering $\mleq$ is \emph{strong} if for every $i,j,k\in V(H)$ with $ij,ik\in E(H)$ and $j\mleq k$, any neighbor of $j$ that comes after $i$ in the ordering is also a neighbor of $k$; in other words, $\{\ell \in N_H(j) \mid i \mleq\ell\} \subseteq \{\ell \in N_H(k) \mid i \mleq\ell\} $. 
	The \emph{width} of an $H$-ordering $\mleq$ is defined as
	$\max\limits_{ij\in E(H)} \min\left( |\{\ell \in N_H(j) \mid i \mleq\ell\}|,
	|\{\ell \in N_H(i) \mid j \mleq\ell\}| \right)$.
	The \emph{degeneracy} of an $H$-ordering $\mleq$ is defined as
	$\max\limits_{i\in V(H)} |\{\ell \in N_H(i) \mid i \mleq\ell\}|$.
\end{definition}

Note that for any $H$-ordering $\mleq$, the width of $\mleq$ is upper-bounded by its degeneracy $d$, as for each $ij\in E(H)$, without loss of generality $i \mleq j$, we have that
$$\min\left(|\{\ell \in N_H(j) \mid i \mleq\ell\}|, |\{\ell \in N_H(i) \mid j \mleq\ell\}|\right) \leq |\{\ell \in N_H(i) \mid j \mleq\ell\}| \leq |\{\ell \in N_H(i) \mid i \mleq\ell\}| \leq d.$$
If $H$ is bipartite, then it can be observed that the definitions of a strong ordering and width do not depend on the relative ordering in $\mleq$ of vertices on different sides of the bipartition.

\begin{lemma}\label{lem:invariant}
Let $M$ be a matrix and let $\mleq$ be a strong $H$-ordering of width $b$, for some completion $H$ of $G_M$. Consider how the matrix $M$ is modified throughout Algorithm~\ref{algo:gaussian} applied on $M$ with order $\mleq$. Then the following invariant is maintained: throughout the algorithm, $H$ is always a completion of $G_M$. Furthermore, the output matrix $L$ has at most $|E(H)|+|V(H)|$ non-zero entries.
\end{lemma}
\begin{proof}
Zero entries of $M$ can become non-zero only in the Entry Update step of Algorithm~\ref{algo:gaussian}.
Specifically, an entry $M[k,\ell]$ can become non-zero  when $M[j,i],M[k,i],M[j,\ell]$ were already non-zero, $i\mleq\ell$ (by choice of $i$), and $j\mleq k$ (by choice of $j$), for certain rows and columns $i,j,k,\ell$.
Before this step, by inductive assumption the invariant held and thus $ij,ik\in E(H)$ and $\ell$ is a neighbor of $j$ in $H$.
Therefore, the strongness of the $H$-ordering $\mleq$ implies that $\ell$ is also a neighbor of $k$ in $H$, that is $k\ell \in E(H)$.
Thus, the invariant is maintained after each step.

To bound the number of non-zero entries of $L$, observe that every non-diagonal entry $L[k,j]$ for rows $k,j$ of $M$ is filled only when there is an edge $i_j k\in E(H)$, where $i_j$ is the column together with which $j$ was removed.
Hence the number of non-zero entries in each column $j$ of $L$ is bounded by $1$ plus the number of rows adjacent to $i_j$ in $H$ (or just 1 if $i_j$ does not exist for $j$).
Since the column $i_j$ is different for different $j$, the total number of non-zero entries in $L$ is bounded by $|E(H)|+|V(H)|$.
\end{proof}

\paragraph*{Implementation.} We now show how to implement Algorithm~\ref{algo:gaussian} on a RAM machine so that the total number of field operations (arithmetic, comparison and assignment on matrix values) and time used for looping through rows and columns is small.
In the following, we consider the RAM model with $\Omega(\log (n))$-bit registers (on input length $n$) and a second type of registers for storing field values, with an oracle that performs field operations in constant time, given positions of registers containing the field values.
The input matrix $M$ is given in any form that allows to enumerate non-zero entries (as triples indicating a row $j$, column $i$ and the position of a register containing $M[j,i]$) in linear time.
The ordering $\mleq$ is given in any form that allows comparison in constant time and enumeration of columns in the order in linear time (otherwise we need $\Oh(|V(H)|\log |V(H)|)$ additional time to sort the columns).
The graph $H$ is given in any form that allows to enumerate edges in linear time. 

\begin{lemma}\label{lem:gaussian}
Let $M$ be a matrix and let $\mleq$ be a strong $H$-ordering of width $b$, for some completion $H$ of the bipartite graph of $M$.
Then Algorithm~\ref{algo:gaussian} on $M$ with pivoting order $\mleq$ makes $\Oh(|E(H)|\cdot b + |V(H)|)$ field operations.

Furthermore, if $H$ and a list of columns sorted with $\mleq$ is given, then all other operations can be done in time $\Oh(|E(H)|\cdot b+|V(H)|)$ on a RAM machine.
\end{lemma}
\begin{proof}
By enumerating all edges of $H$, we can construct adjacency lists for the graph in $\Oh(|E(H)|)$ time.
For every vertex $i$ of $H$, we store the following:
\begin{itemize}
\item two auxiliary values $a[i],a'[i]$, set initially to $-1$;
\item a bit remembering whether $i$ was removed from $\Cols$ or $\Rows$,
\item and we additionally compute an array $p[i]$ containing its last $b$ neighbors in $\mleq$ order, not necessarily sorted.
\end{itemize}
Note that all the arrays $p[i]$ can be constructed in a total of $\Oh(\sum_{i\in V(H)} |N_H[i]| \cdot b) = \Oh(|E(H)| \cdot b)$ time, by finding the last $b$ elements of the adjacency list of $i$ in time $\Oh(|N_H[i]| \cdot b)$, for each $i\in V(H)$.

For every edge $ij$ of $H$, we store a record containing:
\begin{itemize}
\item the position of the registers containing values $a[i],a'[i]$, $a[j],a'[j]$ and $M[j,i]$, and 
\item an auxiliary value $a''[j,i]$, set initially to $-1$.
\end{itemize}
The record is pointed to by each occurrence of the edge in the lists and $p$-arrays of $i$ and $j$ (just as one would store an edge's weight).

\begin{claim}\label{cl:shallow}
Every edge $ij$ of $H$ occurs in at least one of the arrays $p[i]$ or $p[j]$. 
\end{claim}
\begin{proof} 
By the definition of the width of ordering $\mleq$, either $|\{\ell \in N_H(j) \mid i \mleq\ell\}| \leq b$ or 
$|\{\ell \in N_H(i) \mid j \mleq\ell\}| \leq b$.
In the first case, $i$ can be found among the last $b$ neighbors of $j$, symmetrically in the second case.
\cqed\end{proof}

Note that by Claim~\ref{cl:shallow}, for each given edge $ij$ we can access the record of $ij$ in time $\Oh(b)$, by iterating through $p[i]$ and $p[j]$. 

Let us consider the number of basic operations (field and constant-time RAM operations) executed throughout the algorithm.
The outer-most loop executes one iteration for each column $i$ in the $\mleq$ order.
Since a sorted list is given, their enumeration takes $\Oh(|V(H)|)$ time.
In every iteration of the outer-most loop, the rows with non-zero intersection with $i$ are neighbors of $i$ in the bipartite graph of $M$ and hence in $H$, by Lemma~\ref{lem:invariant}.
Therefore, they can be enumerated  by following the adjacency list for $i$ and comparing the corresponding entries with zero, for a total of at most $\Oh(\sum_{i\in V(H)} |N_H[i]|)=\Oh(|V(H)|+|E(H)|)$ basic operations.
This also suffices for making comparisons to find the earliest row $j$ among them.

For any fixed $i$, after choosing $j$, the inner loops perform the Entry Update step for every intersection of a row $k\in R$ and a column $\ell \in C$, where $R$ is the set of rows in $\Rows$ having non-zero intersection with $i$ and $C$ is the set of columns in $\Cols$ having non-zero intersection with $j$.
By Lemma~\ref{lem:invariant}, $R\subseteq N_H(i)$ and $C\subseteq N_H(j)$. We find $R$ and $C$ by iterating over $N_H(i)$ to find $R$, and iterating over $N_H(j)$ to find $C$. Note that the time spent on all these iterations amortizes to $\Oh(\sum_{i\in V(H)} |N_H(i)|)=\Oh(|E(H)|)$.

To access all the values $M[k,\ell]$ occurring at intersections $k\in R$ and $\ell \in C$, we iterate over the arrays $p[k]$ and $p[\ell]$ of each $k\in R$ and each $\ell\in C$.
By Claim~\ref{cl:shallow}, each edge $k\ell$ for $k\in R,\ell\in C$ will be in at least one of these arrays (in particular $|R|\cdot|C| \leq b \cdot (|R|+|C|)$).
Thus the total time used on Entry Update steps will amount to $\Oh(b \cdot(|R|+|C|))$ basic operations.
More precisely, for each $k\in R$, we set $a[k] := M[k,i]$ and $a'[k] := i$.
Similarly, for each $\ell \in C$, we set $a[\ell] := M[j,\ell]$ and $a'[\ell] := i$.
Then, we iterate over all edges $k\ell$ with $k\in R$ and $\ell \in C$ using the $p$-arrays, and perform the Entry Update step for each such edge.
Given $k$ and $\ell$, we can check that they are indeed in $R$ and $C$ by testing $a'[k]=i$ and $a'[\ell]=i$.
We can ensure that the update is not performed twice on the same entry by setting $a''[k,\ell]:=i$ when the first update is performed, and then not performing it again when value $i$ is seen in $a''[k,\ell]$.
Values $M[k,i]$ and $M[j,\ell]$ can be accessed in constant time via $a[k]$ and $a[\ell]$.

Since $|R|\leq |N_H(i)|$ and $|C|\leq |N_H(j)|$, the total number of basic operations used for the Entry Update steps is $\Oh(b\cdot (|N_H(i)|+|N_H(j)|))$. Since $i$ and $j$ are afterwards removed from \Cols and \Rows, this amortizes to $\Oh(b\cdot \sum_{i\in V(H)} |N_H(i)|)=\Oh(b\cdot |E(H)|)$ basic operations in total. 
Therefore, the total number of basic operations made throughout the algorithm is bounded by $\Oh(|E(H)| \cdot b + |V(H)|)$.
\end{proof}

\subsection{Orderings for path and tree-partition decompositions}\label{sec:pw-elim}
In this section we show that small path or tree-partition decompositions of the bipartite graph associated with a matrix can be used to find a completion with a strong ordering of small width.
In both cases it is enough to complete the graph to a maximal graph admitting the same decomposition and take a natural ordering (corresponding to ``forget nodes'' -- the rightmost or topmost bags that contain each vertex).
	
\begin{lemma}\label{lem:pw-ordering}
Given a matrix $M$ and a path decomposition of width $b$ of the bipartite graph $G=G_M$ of $M$, one can construct a completion $H$ of $G$ with at most $2b\cdot |V(G)|$ edges and list a strong $H$-ordering of degeneracy (and hence width) at most $b$, in time $\Oh(b \cdot |V(G)|)$.
\end{lemma}
\begin{proof}
Consider a path decomposition of $G$ with consecutive bags $B_1,B_2,\dots,B_q$; since we can make it a clean decomposition in linear time, we can assume that $q\leq n$.
For every vertex $v$ of $G$, let $B_{b(v)}, B_{e(v)}$ be the first and last bag containing $v$, respectively. For $i=1,2,\ldots,n$, let $B_i'$ be the set of all the vertices $v$ with $e(v)=i$.
Let $H$ be the graph obtained from $G$ by adding edges between any two vertices in the same bag $B_i$ (that is, sets $B_i$ become cliques in $H$).
The graph $H$ still has pathwidth $b$ so by Lemma~\ref{lem:tw-edges} it has at most $b \cdot |V(G)|$ edges.
Let $\mleq$ be any ordering that places all vertices in $B_i'$ before all vertices $B_j'$, for $i<j$ (vertices within one set $B_i'$ can be ordered arbitrarily);
that is, $e(u) < e(v)$ implies $u \mleq v$.
It is straightforward to perform the construction in time $\Oh(b \cdot |V(G)|)$.
We claim $\mleq$ is a strong $H$-ordering of degeneracy $b$.

To show this, first observe that $uv \in E(H)$ ($u$ and $v$ were in a common bag of the decomposition) if and only if $b(u) \leq e(v)$ and $b(v) \leq e(u)$ --- if $uv\in E(H)$, the vertices were in a common bag and the implication is clear, while in the other case either $b(u),e(u) < b(v)$ or $b(v),e(v)< b(u)$, giving the converse.
To check strongness, let $i,j,k\in V(H)$ be such that $ij,ik\in E(H)$ and $j\mleq k$.
Then $ik\in E(H)$ implies $b(k) \leq e(i)$ and $j\mleq k$ implies $e(j) \leq e(k)$.
Let $\ell$ be any neighbor of $j$ that comes after $i$.
Then $\ell j \in E(H)$ implies $b(\ell) \leq e(j)$ and $i\mleq \ell$ implies $e(i) \leq e(\ell)$.
Together, we have $b(\ell) \leq e(j) \leq  e(k)$ and $b(k) \leq e(i) \leq e(\ell)$, hence $\ell k \in E(H)$, concluding the proof of strongness.

To bound the degeneracy of $\mleq$, for each $i\in V(H)$ we want to bound the number of neighbors $\ell$ of $i$ with $i \mleq\ell$.
Such a neighbor must satisfy $b(\ell)\leq e(i)$ and $e(i) \leq e(\ell)$.
By the properties of a decomposition, $\ell$ must be contained in all bags from $B_{b(\ell)}$ to $B_{e(\ell)}$, hence both $i$ and $\ell$ are contained in the bag $B_{e(i)}$.
Therefore, the number of such neighbors $\ell$ is bounded by $|B_{e(i)}\setminus \{i\}|\leq b$ for each $i\in V(H)$, which shows degeneracy is at most $b$.
\end{proof}

\begin{lemma}\label{lem:stw-ordering}
	Given a matrix $M$ and a tree-partition decomposition of width $b$ of the bipartite graph $G$ of $M$, one can construct a completion $H$ of $G$ with at most $b\cdot |V(G)|$ edges and list a strong $H$-ordering of degeneracy (and hence width) at most $2b$, in time $\Oh(b \cdot |V(G)|)$.
\end{lemma}
\begin{proof}
	Let $(\Tt,\{B_t\}_{t\in V(\Tt)})$ be the given tree-partition decomposition of $G$.
	Let $H$ be the graph obtained from $H$ by adding all edges between vertices in the same or in adjacent (in $\Tt$) bags.
	The graph $H$ still has tree-partition width $b$ and hence at most $2b \cdot |V(G)|$ edges, by Corollary~\ref{cor:tpw-edges}. For a vertex $i\in V(H)$, by $t(i)$ we denote the node of $\Tt$ whose bag contains $i$.

	We root $\Tt$ arbitrarily, which imposes an ancestor-descendant relation on the nodes of $\Tt$.
	Let $\mleq$ be any ordering that goes `upward' the decomposition, that is, places all vertices of $B_t$ after all vertices of $B_{t'}$ for any $t$ and its descendant $t'$ in $\Tt$. 
	It is straightforward to perform the construction of any such $\mleq$ in $\Oh(b \cdot |V(G)|)$.
	We claim $\mleq$ is a strong $H$-ordering of degeneracy $2b$.
	
	The bound on degeneracy follows from the fact that the neighbors of a vertex $i$ in $H$ occurring later in the ordering $\mleq$ must be either in the same bag as $i$, or in the parent bag of the bag containing $i$.
	To show strongness, let $i,j,k\in V(H)$ be such that $ij,ik\in E(H)$ and $j\mleq k$. 
	Let $\ell$ be any neighbor of $j$ that comes after $i$ in $\mleq$.
	We want to show that $\ell$ is a neighbor of $k$ too.
	If it is not, then $t(\ell)\neq t(i)$ (as otherwise $N_H[\ell]=N_H[i]\ni k$), and similarly $t(k)\neq t(j)$ (as otherwise $N_H[k]=N_H[j]\ni \ell$).
	Since $j$ and $i$ are adjacent, $t(i)$ and $t(j)$ are either equal or adjacent (in $\Tt$).
	
	If $t(i)$ is a child of $t(j)$, then since $t(k)$ is either equal or adjacent to $t(i)$ (due to $ik\in E(H)$) and it is not a descendant of $t(j)$ (by $j \mleq k$), it must be equal to $t(j)$, contradicting the above inequalities.
	
	If $t(j)$ is a child of $t(i)$, then since $t(\ell)$ is either equal or adjacent to $t(j)$ (due to $j\ell\in E(H)$) and it is not a descendant of $t(i)$ (by $i \mleq \ell$), it must be equal to $t(i)$, contradicting the above inequalities.
	
	If $t(i)=t(j)$, then $t(k)$ is either equal or adjacent to $t(i)=t(j)$ (by $ik\in E(H)$), they cannot be equal (by the above inequalities) and $t(k)$ is not a child of $t(j)$ (by $j \mleq k$), hence $t(k)$ must be the parent of $t(i)=t(j)$.
	Similarly, $t(\ell)$ is either equal or adjacent to $t(j)=t(i)$ (by $j\ell \in E(H)$), they cannot be equal (by the above inequalities), and $t(\ell)$ is not a child of $t(i)$ (by $i\mleq \ell$), hence $t(\ell$) must also be the parent of $t(j)=t(i)$, implying $t(\ell)=t(k)$. Hence in any case $\ell$ is a neighbor of $k$ too, proving strongness of the $H$-ordering $\mleq$.
\end{proof}

Given Lemmas~\ref{lem:pw-ordering} and \ref{lem:stw-ordering}, from an $n\times m$ matrix $M$ and a path- or tree-partition- decomposition of $G_M$ of width $b$, we can construct a completion $H$ of $G_M$ with at most $(n+m)\cdot b$ edges and a strong $H$-ordering of width at most $2b$.
Therefore, Gaussian elimination can be performed using $\Oh((n+m)\cdot b^2)$ field operations and time by Lemma~\ref{lem:gaussian}, yielding matrices with at most $(n+m)\cdot b$ non-zero entries by Lemma~\ref{lem:invariant}.
Together with Lemma~\ref{lem:usingGaussian}, this concludes the proof of Theorem~\ref{thm:pw-det}.

\restatepwdet*

It is tempting to try to perform the same construction as in Lemmas~\ref{lem:pw-ordering} and~\ref{lem:stw-ordering} also for standard tree decompositions that correspond to treewidth. That is, complete the graph to a chordal graph $H$ according to the given tree decomposition, root the decomposition in an arbitrary bag, and order the vertices in a bottom-up manner according to their forget nodes (i.e., highest nodes of the tree containing them). Unfortunately, it is not hard to construct an example showing that this construction {\em{does not}} yield a strong ordering. 
In fact, there are chordal graph that admit no strong ordering~\cite{Dragan00}.
For this reason, in the next section we show how to circumvent this difficulty by reducing the case of tree decompositions to tree-partition decompositions using the vertex splitting technique.

%% file: figureElimination.tex
\begin{tikzpicture}[scale=0.22]
	\node (i) at (8.1,0.4) {$i$};
	\draw (7.5,-0.5) rectangle (8.5,-19.5);
	\node (l) at (11.1,0.4) {$\ell$};
	\draw[draw=none,fill=blue!60!gray!10!white] (10.5,-0.5) rectangle (11.5,-19.5);
	\node at (9.6,0.4) {$\mleq$};

	\draw (0.5,-0.5) rectangle (24.5,-19.5);

	\node (j) at (-0.9,-9) {$j$};
	\draw (0.5,-8.5) rectangle (24.5,-9.5);
	\node (k) at (-0.9,-16) {$k$};
	\draw (0.5,-15.5) rectangle (24.5,-16.5);
	\draw[->] (j) to[bend right] (k);
	\node[rotate=-90] at (-0.9,-12.5) {$\mleq$};
	\draw[gray,->] ( 8,-9) to[bend right] ( 7.4,-15.4);
	\draw[gray,->] (11,-9) to[bend right] (10.8,-15.4);

	\draw[fill] (8,- 9) circle (0.3); 
	\draw[fill] (8,-13) circle (0.3); 
	\draw[fill] (8,-14) circle (0.3); 
	\draw[fill] (8,-16) circle (0.3); 
	\draw[fill] (8,-18) circle (0.3); 

	\draw[fill] (11,-9) circle (0.3); 
	\draw[fill] (13,-9) circle (0.3); 
	\draw[fill] (17,-9) circle (0.3); 
	\draw[fill] (18,-9) circle (0.3); 
	\draw[fill] (19,-9) circle (0.3); 
	\draw[fill] (23,-9) circle (0.3); 

	\draw[fill] (1,-1) circle (0.3); 
	\draw[fill] (2,-2) circle (0.3); 
	\draw[fill] (4,-3) circle (0.3); 
	\draw[fill] (3,-5) circle (0.3); 
	\draw[fill] (5,-4) circle (0.3); 
	\draw[fill] (6,-15) circle (0.3); 
	\draw[fill] (7,-7) circle (0.3); 
	\draw[fill] (10,-8) circle (0.3); 
	\draw[fill] (14,-8) circle (0.3); 
	\draw[fill] (10,-8) circle (0.3); 
	\draw[fill] (10,-10) circle (0.3); 
	\draw[fill] (18,-10) circle (0.3); 
	\draw[fill] (19,-10) circle (0.3); 
	\draw[fill] (22,-10) circle (0.3); 
	\draw[fill] ( 9,-11) circle (0.3); 
	\draw[fill] (24,-11) circle (0.3); 
	\draw[fill] (15,-11) circle (0.3); 
	\draw[fill] (14,-12) circle (0.3); 
	\draw[fill] (11,-13) circle (0.3); 
	\draw[fill] (14,-13) circle (0.3); 
	\draw[fill] (17,-13) circle (0.3); 
	\draw[fill] (18,-13) circle (0.3); 
	\draw[fill] (19,-13) circle (0.3); 
	\draw[fill] (23,-13) circle (0.3); 
	\draw[fill] (24,-13) circle (0.3); 
	\draw[fill] (11,-14) circle (0.3); 
	\draw[fill] (13,-14) circle (0.3); 
	\draw[fill] (17,-14) circle (0.3); 
	\draw[fill] (19,-14) circle (0.3); 
	\draw[fill] (23,-14) circle (0.3); 

	\draw[    ] (11,-16) circle (0.3); 
	\draw[fill] (13,-16) circle (0.3); 
	\draw[    ] (17,-16) circle (0.3); 
	\draw[    ] (18,-16) circle (0.3); 
	\draw[fill] (19,-16) circle (0.3); 
	\draw[fill] (23,-16) circle (0.3); 
	\draw[fill] (12,-16) circle (0.3); 
	\draw[fill] (22,-16) circle (0.3); 
	\draw[fill] (24,-17) circle (0.3); 
	\draw[fill] (21,-18) circle (0.3); 
	\draw[fill] (13,-18) circle (0.3); 
	\draw[fill] (9,-19) circle (0.3); 
	\draw[fill] (15,-19) circle (0.3); 
\end{tikzpicture}

%% file: gaussian-tw.tex
\subsection{Vertex splitting for low treewidth matrices}\label{sec:tw-elim}
In this subsection we show how vertex splitting can be used to expand a matrix of small treewidth into an equivalent matrix of small tree-partition width.
This allows us to extend our results to treewidth as well and prove Theorem~\ref{thm:tw-det}, without assuming the existence of a good ordering for the original matrix.

The vertex splitting operation, as described in the introduction, can be used repeatedly to change vertices into arbitrary trees, with original edges moved quite arbitrarily.
While algorithms will not need to perform a series of splittings, as the final outcome in our application can be easily described and constructed directly from a given tree decomposition, we nevertheless define possible outcomes in full generality to perform inductive proofs more easily. 

In this section, each tree decomposition has a tree that is arbitrarily rooted and the set of all children of each node is arbitrarily ordered, so that the following ordering can be defined. The \emph{pre-order} of an ordered tree is the ordering of nodes which places a parent before its children and the children in the same order as defined by the tree.

A {\em{tree-split}} \E of a graph $G$ is an assignment of an ordered tree $\E(v)$ to every vertex $v\in V(G)$ and of a node pair $\E(uv) = (t,t')$ for every edge $uv\in E(G)$, such that $t \in V(\E(u))$ and $t' \in V(\E(v))$.
A rooted tree decomposition $(\Tt,\{B_t\}_{t\in V(\Tt)})$ (with an ordered tree $\Tt$) of a graph $G$ gives rise to a tree-split $\E(\Tt)$ as follows: for $v\in V(G)$, $\E(v)$ is the subtree of $\Tt$ induced by those nodes whose bags contain $v$; for $uv \in E(G)$, $\E(uv)=(t,t)$, where $t$ is the topmost node of $\treedecomp$ whose bag contains both $u$ and $v$ (it is easy to see that there is always a unique such node).
See Figure~\ref{fig:splitting} for an example.

For a matrix $M$ and a tree-split \E of $G=G_M$, we define below the $\E$-split of $M$, denoted $M_{\E}$.
We later show that this operation preserves the determinant up to sign, for example.
We will denote the $\E$-split of $M$ corresponding to a tree decomposition $\T$ of the bipartite graph of $M$ simply as $M_{\T}$ -- we aim to show that $M_{\T}$ preserves the algebraic properties of $M$ and strengthens the structure given by $\T$ to that of a tree-partition decomposition.
The ordering and sign choices in the definition are only needed to preserve the sign of the determinant.

\begin{definition}[$\E$-split of $M$]
Let $M$ be a matrix with rows $r_1,\dots,r_n$, columns $c_1,\dots,c_m$ and bipartite graph $G=G_M$, and let $\E$ be a tree-split of $G$.
The matrix $M_\E$ has the following rows, in order: for every row $r_i$ of the original matrix (in original order) we have a row indexed with the pair $(r_i,t)$, where $t$ is the root of $\E(r_i)$.
After this, for every original row $r_i$ we have consecutively, for every non-root node $t$ of the tree $\E(r_i)$, a row indexed with the pair $(r_i,t)$ (different $t$ occurring according to the pre-order of the tree).
Then, for every original column $c_i$ (in original order), we have consecutively for every edge $tt'$ of the tree $\E(c_i)$ a row indexed with the pair $(c_i,tt')$ (different $tt'$ occurring according to the pre-order of the lower node in the tree $\E(c_i)$).
The columns of $M_\E$ are defined symmetrically, that is, they are indexed with $(c_i,t)$ and $(r_i,tt')$.

We define the entries of $M_\E$.
For every non-zero entry of the original matrix, that is, every edge $r_i c_j$ of its bipartite graph, set $M_\E[(r_i,t),(c_j,t')] = M[r_i,c_j]$, 
where $(t, t') = \E(r_i c_j)$.
For each row indexed as $(c_i,tt')$, with $t$ being the parent of $t'$ in $\E(c_i)$, set $M'[(c_i,tt'),(c_i,t')]=-M'[(c_i,tt'),(c_i,t)]=(-1)^{n'}$, where $n'$ is the number of rows preceding $(c_i,t')$.
Symmetrically, for each column indexed as $(r_i,tt')$, with $t$ the parent of $t'$ in $\E(r_i)$, set $M'[(r_i,t'),(r_i,tt')]=-M'[(r_i,t),(c_i,tt')]=\pm 1$ analogously.
We set all other entries to 0.
This concludes the definition.
\end{definition}

\begin{figure}[H]
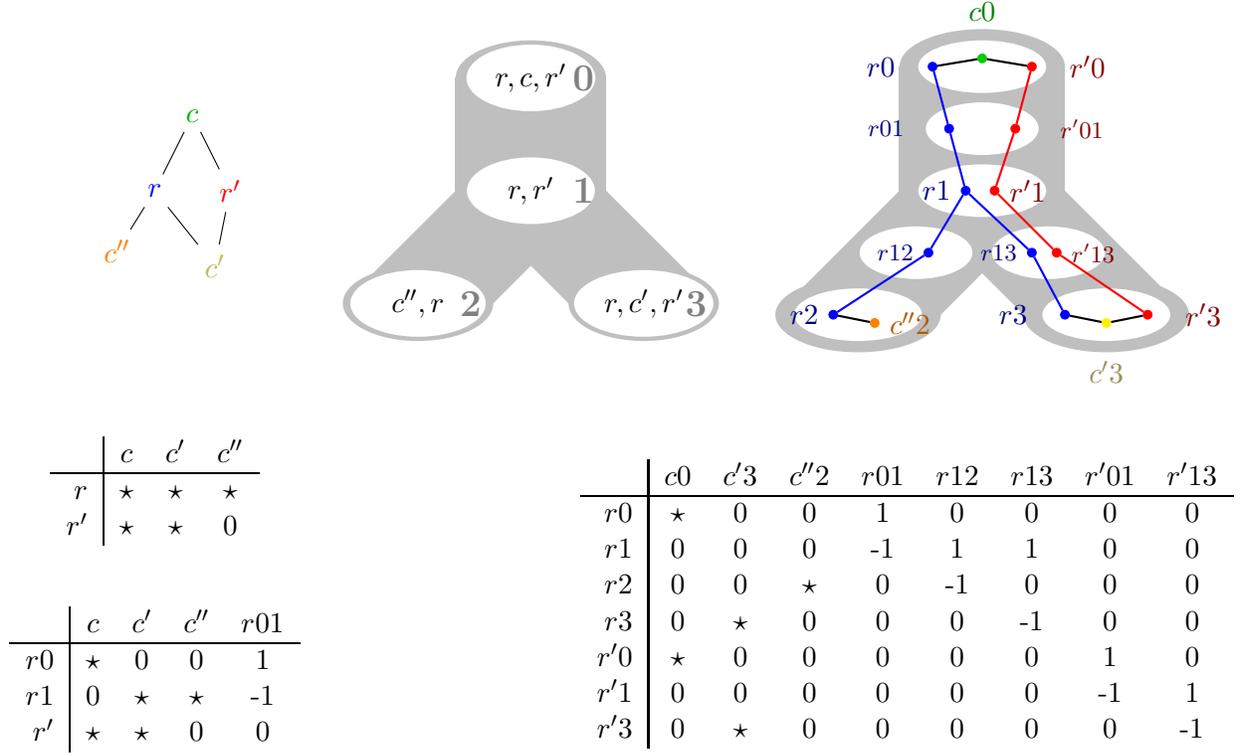

	\centering
	\include{figureSplitting}
	\tikzstyle{v}=[circle,fill=black,draw=black!75,inner sep=0pt,minimum size=0.3em]
	\caption{On the left: a matrix $M$, its bipartite graph $G_M$, and the matrix after splitting a row once. In the middle: a tree decomposition of $G_M$ with nodes $0,1,2,3$. On the right: the matrix after splitting according to this decomposition, and its graph in a tree-partition decomposition.}
	\label{fig:splitting}
\end{figure}

For a tree-split $\E$ of a graph $G$, we write $\|\E\|$ for the total number of edges in all trees $\E(v)$, for $v\in V(G)$.
Note that if $\E$ is a tree-split of the bipartite graph of an $n\times m$ matrix $M$, then $M_{\E}$ has exactly $n+\|\E\|$ rows and $m+\|\E\|$ columns. 
For a set of rows (analogously for columns) $I$, we write $I_{\E}$ for the set of all rows of $M_{\E}$, except those of the form $(v,t)$ where $v$ is a row of $M$ not in $I$ and $t$ is the root of $\E(v)$.
In other words, $I_{\E}$ is obtained from $I$ by taking rows with the same positions, and adding all the last $\|E\|$ rows of $M_{\E}$.
In particular, $|I_{\E}| = |I| + \|\E\|$.

\begin{lemma}\label{lem:split}
	Let $\E$ be a tree-split of the bipartite graph $G$ of an $n\times m$ matrix $M$.
	Then $\rk M = \rk M_{\E} - \|\E\|$
	and for any sets $I$ and $J$ of rows and columns of $M$ of equal size, 
$$\det[M]_{I,J} = (-1)^{\|\E\| \cdot (n+m)} \cdot \det[M_{\E}]_{I_{\E},J_{\E}}.$$
In particular $\det M_{\E} = \det M$, if $n=m$. 
\end{lemma}

To prove the above lemma, we need the following definition and claim to perform an inductive step.
For a tree-split $\E$ of a graph $G$ and two adjacent nodes $a,b$ of a tree $\E(v)$ assigned to some vertex $v \in V(G)$, define the \emph{contracted} tree-split $\E_{v/ab}$ as $\E$ with the two nodes $a,b$ identified in $\E(v)$  (contracting the tree edge connecting $a$ and $b$) and identified in any pair $\E(vw)$ that contains them.

\begin{claim}\label{claim:splitStep}
	Let $\E$ be a tree-split of the bipartite graph $G=G_M$ of a matrix $M$.
	Let $(v,tt')$ be the last column or the last row of $M_{\E}$ (so $t'$ is the last node in pre-order of the tree $\E(v)$ and $t$ is its parent).
	Then $\rk M_{\E_{v/tt'}} = \rk M_{\E} - 1$ and if $I,J$ are any sets of rows and columns of $M$ of equal size, then $\det[M_{\E_{v/tt'}}]_{I_{\E_{v/tt'}},J_{\E_{v/tt'}}} = (-1)^{n'}\cdot \det[M_{\E}]_{I_{\E},J_{\E}}$, where $n'$ is the number of rows (if $v$ is a row) or columns (if $v$ is a column) of $M_{\E_{v/tt'}}$.
\end{claim}
\begin{proof}
	Without loss of generality assume $v$ is a row of $M$.
	Note that the signs in the definition of $M_{\E}$ were chosen so that removing the last column or last row does not change any of them.
	Hence the matrices $M_{\E}$ and $M_{\E_{v/tt'}}$ only differ at rows $(v,t)$, $(v,t')$, and column $(v,tt')$. 
	Observe that the row $(v,t)$ of $M_{\E_{v/tt'}}$ is obtained by adding rows $(v,t')$ and $(v,t)$ of $M_{\E}$ and deleting column $(v,tt')$.
	Thus adding row $(v,t')$ to row $(v,t)$ of $M_{\E}$ yields
	a matrix equal to $M_{\E_{v/tt'}}$, but with an additional column $(v,tt')$ and an additional row $(v,t')$.
	After this row operation, the column $(v,tt')$ has only one non-zero entry $\sigma=\pm 1$ at the intersection with row $(v,t')$ (the other has just been canceled).
	Hence all other entries of this additional row can be eliminated using this entry.
	This makes $M_{\E}[(v,t'),(v,tt')]=\sigma$ the only non-zero entry in its row and column (after applying the above row and column operations).
	Hence $\rk M_{\E} = \rk M_{\E_{v/tt'}} + 1$.
	
	To check minor determinants, consider any 
	sets $I,J$ of rows and columns of $M$ of equal size.
	The sets $I_{\E}$ and $J_{\E}$ always contain $(v,t')$ and $(v,tt')$,
	so the above row and column operations have the same effect after deleting rows and columns outside those sets, and do not change the determinant.
	Moving row $(v,t')$ to the last position multiplies the determinant by $(-1)^{n_1}$, where $n_1$ is the number of rows below it -- since $I_{\E}$ contains all rows below $(v,t')$, this is independent of $I$ and $J$ and we can count rows in $M_{\E}$ just as well.
	Then, deleting this last row and the last column (whose only non-zero entry is their intersection $\sigma$) multiplies the determinant by $\sigma=(-1)^{n_2}$, where $n_2$ was defined to count the rows originally above $(v,t')$ in $M_{\E}$.
	This deletion yields a matrix equal to $M_{\E_{v/tt'}}$, hence $\det M_{\E} = (-1)^{n'} \cdot \det M_{\E_{v/tt'}}$, where $n'=n_1+n_2$ counts all rows of $M_{\E}$ except the deleted one, which is equal to the number of all rows in $M_{\E_{v/tt'}}$.	
\cqed\end{proof}

\begin{proof}[Proof of Lemma~\ref{lem:split}]
	Observe that if $\E'$ is a trivial tree-split assigning a single node tree $\E(v)$ to every vertex $v\in V(G)$, then $M_{\E'}$ defines the same matrix as $M$.
	Hence $M$ can be obtained from $M_{\E}$ by repeatedly contracting the edge corresponding to the last row or the last column.
	By Claim~\ref{claim:splitStep}, the rank decreases by exactly one with every contraction, so in total it decreases by $\|\E\|$.
	
	Similarly, each minor's determinant changes sign $i$ times with every contraction, where $i$ counts the rows or columns of the matrix after contraction.
	Hence in total, the number of sign changes is $\sum_{i=n+\|\E\|-1}^{n} i$ (for contractions corresponding to rows) plus $\sum_{i=m+\|\E\|-1}^{m} i$ (for columns), which is equal to $\frac{\|\E\|\cdot(n+\|\E\|-1+n)}{2}+\frac{\|\E\| \cdot (m+\|\E\|-1+m)}{2} = \|\E\|\cdot (n+m+\|\E\|-1) \equiv \|\E\|\cdot (n+m)\quad (\mbox{mod }2)$.
\end{proof}

The explicit construction of the split of a matrix allows us to easily bound its size, number of non-zero entries in each row and in total.
Before this, to optimize these parameters, we need the following easy adaption of a standard lemma about constructing so called `nice tree decompositions' (see \cite{BodlaenderBL13}).

\begin{lemma}\label{lem:niceDecomp}
	Given a tree decomposition of a graph $G$ of width $b$, one can find in time $\Oh(bn)$ a tree decomposition of $G$ of width $b$ with a rooted tree of at most $5n$ nodes, each with at most two children.
	Furthermore, for every node $t$ of the decomposition, there are at most $b$ edges $uv\in E(G)$ such that $t$ is topmost node whose bag contains both $u$ and $v$.
\end{lemma}
\begin{proof}
	Root the given tree decomposition in an arbitrary node. For $v\in V(G)$, let $t(v)$ be the topmost node whose bag contains $v$ -- there is such a bag by properties of tree decompositions (the so called \emph{forget} bag for $v$).	
	Bodlaender, Bonsma and Lokshtanov~\cite[Lemma~6]{BodlaenderBL13} describe how to transform a tree decomposition in time $\Oh(n\cdot b)$ to a `nice tree decomposition', which in particular has the following properties: it has at most $4|V(G)|$ nodes, each with at most two children, and furthermore $t(v)$ has at most one child for each $v\in V(G)$, and for $u\neq v$, either $t(u)\neq t(v)$ or $t(u)=t(v)$ is the root of the decomposition tree.
	To remedy this last possibility, if the root node's bag is $B=\{b_1,\dots,b_{\ell}\}$, add atop of it a path of nodes with bags $\{b_1,\dots,b_i\}$ for $i=\ell,\ell-1,\dots,1,0$, rooted at the last, empty bag.
	This adds at most $\ell\leq |V(G)|$ nodes to the decomposition tree.
	Now $t(v)\neq t(u)$ for all vertex pairs $u\neq v \in V(G)$.
		
	Consider now any edge $uv \in E(G)$.
	Since the sets of nodes whose bags contain $u$ and those who contain $v$ induce connected subtrees of the decomposition tree, their intersection is also a connected subtree whose topmost bag cannot be a descendant of both $t(u)$ and $t(v)$.
	That is, the topmost node whose bag contains both $u$ and $v$ is either $t(u)$ or $t(v)$.
	Therefore, if we assign each edge to the topmost bag that contains both its endpoints, then for each node $t$, the edges assigned to it must be incident to $v$, where $v\in V(G)$ is such that  $t=t(v)$.
	Since such edges have both endpoints in the bag of $t$, there can be at most $b$ of them. 
\end{proof}

\begin{lemma}\label{lem:twsplitting}
	Let $M$ be an $n\times m$ matrix with bipartite graph $G=G_M$.
	Let $\treedecomp$ be a tree decomposition of $G$ of width $b$ obtained from Lemma~\ref{lem:niceDecomp}.
	Then $M_\treedecomp$ has the following properties, for some $N=\Oh(b\cdot (n+m))$:
	\begin{enumerate}[(a)]
		\item\label{pr:dim} $M_\treedecomp$ has $n+N$ rows and $m+N$ columns,
		\item\label{pr:sparse} Every row and column of $M_\treedecomp$ has at most $b+3$ non-zero entries, and $M_\treedecomp$ has at most $|E(G)|+4N$ such entries in total,
		\item\label{pr:tptw} The bipartite graph of $M_\treedecomp$ has a tree-partition decomposition of width $b$, with a tree that is the 1-subdivision of the tree of $\treedecomp$,
		\item\label{pr:time} $M_\treedecomp$ and the decomposition can be constructed in time $\Oh(N)$,
		\item\label{pr:rank} $\rk M = \rk M_\treedecomp - N$,
		\item\label{pr:minor} $\det[M]_{I,J} = (-1)^{N \cdot (n+m)} \cdot \det[M_{\treedecomp}]_{I',J'}$ for any sets $I$ and $J$ of rows and columns of $M$ of equal size,
		where $I'$ ($J'$) is obtained from $I$ ($J$) by adding the last $N$ rows (columns) of $M_\treedecomp$ to it,\\
		(in particular $\det M_{\treedecomp} = \det M$, if $n=m$),
	\end{enumerate}
\end{lemma}
\begin{proof}
	Let $\E$ be the tree-split corresponding to $\treedecomp$, let $M_\treedecomp=M_\E$ and let $N=\|\E\|$.
	Property~\eqref{pr:dim} follows directly from the definition of $M_\E$.	
	Since by definition trees $\E(v)$ are subtrees of $\treedecomp$, and, by the definition of a decomposition's width, at most $b+1$ such subtrees can share each edge of this tree, we have that $N=\|\E\|=\Oh(b \cdot (n+m))$.
	
	Property~\eqref{pr:sparse} follows from the fact that every row of $M_\treedecomp$ is either indexed as $(c,tt')$ (for some column $c$ of $M$) --- in which case it has exactly two non-zero entries --- or it is indexed as $(r,t)$ for some row $r$ of $M$ and some node $t$ of $\treedecomp$ whose bag contains $r$.
	The only non-zero entries of row $(r,t)$ are: $M_\treedecomp[(r,t),(c,t)]$ for edges $rc\in E(G)$ such that  $\E(rc)=(t,t)$;
	and $M_\treedecomp[(r,t),(r,tt')]$ for neighbors $t'$ of $t$ in $\treedecomp$.
	By the guarantees of Lemma~\ref{lem:niceDecomp}, there are at most $b$ edges $rc$ such that $t$ is the topmost bag containing both endpoints, that is, such that $\E(rc)=(t,t)$.
	Furthermore, every node has at most three neighbors in the tree of $\treedecomp$, thus the row $(r,t)$ has at most $b+3$ non-zero entries in total.
	The proof is symmetric for columns.
	Similarly, the total number of non-zero entries can be bounded by $2N$ for entries in rows indexed as $(c,tt')$, $2N$ for entries in columns indexed as $(r,tt')$ and $|E(G)|$ for the remaining entries, which must be of the form
	$M_\treedecomp[(r,t),(c,t)]$ where $rc \in E(G)$ and $(t,t)=\E(rc)$.
	
	Properties~\eqref{pr:tptw} and~\eqref{pr:time} follow easily from the construction:
	the tree-partition decomposition of $M_\treedecomp$'s bipartite graph will have a bag $V_t$ for every node $t$ of $\treedecomp$ and a bag $V_{tt'}$ for every edge $tt'$ of $\treedecomp$ -- these bags are naturally assigned to the nodes of the 1-subdivision of $\treedecomp$'s tree.
	Each bag $V_t$ contains all rows and columns indexed as $(v,t)$, for some $v\in V(G)$ (contained in the bag of $t$ in $\treedecomp$) and each bag $V_{tt'}$ contains all rows and columns indexed as $(v,tt')$, for some $v\in V(G)$ (contained in both bags of $t,t'$ in $\treedecomp$).
	This defines a valid tree-partition decomposition of width at most $b$, as non-zero entries are  either of the form $M_\treedecomp[(r,t),(c,t)]$ and hence its row and column fall into the same bag $V_t$, or of the form $M_\treedecomp[(r,t),(r,tt')]$ (or symmetrically for columns) and hence its row and column fall into adjacent bags $V_t$, $V_{tt'}$.
	
	Properties~\eqref{pr:rank} and~\eqref{pr:minor} follow directly from Lemma~\ref{lem:split}.
\end{proof}

We remark that the construction of Lemma~\ref{lem:twsplitting} actually preserves the symmetry of the matrix, i.e., if $M$ is symmetric then so is $M_\E$ as well. The construction can be also easily adapted to preserve skew-symmetry, if needed.


As a final ingredient for Theorem~\ref{thm:tw-det} we need to provide a generalized $LU$-factorization for the original matrix and retrieve a maximal nonsingular submatrix.

\begin{lemma}\label{lem:twLU}
	Let $M$ be an $n \times m$ matrix over a field $\F$ and let $\E$ be a tree-split of the bipartite graph $G=G_M$.
	Given a $PLUQ$-factorization of $M_\E$ with $N$ non-zero entries in total, a generalized $LU$-factorization of $M$ with at most $N'=N+2\|\E\|+2n+2m$ non-zero entries can be constructed in $\Oh(N')$ time.
	Hence given any vector $r\in \F^m$, the system of linear equation $Mx=r$ can be solved in $\Oh(N')$ additional field operations and time.
\end{lemma}
\begin{proof}
	Define $U_\E$ as the following matrix with $n$ rows corresponding to rows of $M$ and $n+\|\E\|$ columns corresponding to rows of $M_\E$ (indexed as $(r,t)$ or $(c,tt')$).
	The only non-zero entries of $U_\E$ are $U_\E[r,(r,t)]=1$ for rows $r$ of $M$ and nodes $t$ of the tree $\E(r)$.
	Analogously, define $L_\E$ as the following matrix with $m$ columns corresponding to columns of $M$ and $m+\|\E\|$ rows corresponding to columns of $M_\E$ (indexed as $(c,t)$ or $(r,tt')$).
	The only non-zero entries of $L_\E$ are $L_\E[(c,t),c]=1$ for columns $c$ of $M$ and nodes $t$ of the tree $\E(c)$.
	It is straightforward from the definition of $M_\E$ that $M=U_\E M_\E L_\E$.
	
	Let $P,L,U,Q$ define the given $PLUQ$-factorization of $M_\E$.
	Then $M=U_\E P L U Q L_\E$.
	Observe that for any ordering of columns of $U_\E$, the matrix can be given in row-echelon form by ordering rows so that $r_1$ comes before $r_2$ whenever the first column indexed as $(r_1,t)$ (for $t\in\E(r_1)$) comes before the first column indexed as $(r_2,t')$ (for $t'\in\E(r_2)$).
	Hence we can construct from $P$ a permutation $n\times n$ matrix $P'$ such that $U'=P' U_\E P$ is in row-echelon form.
	Similarly we can find an $m\times m$ permutation matrix $Q'$ and a reordering $L'$ of $L_\E$ such that $L'=QL_\E Q'$ is in column-echelon form.
	Then $M=P'^{-1} U' L U L' Q'^{-1}$ gives a generalized $LU$-factorization.
	The number of non-zero entries in $L_\E$ and in $U_\E$ is equal to $n+\|\E\|$ and $m+\|\E\|$, respectively. Hence the total number of non-zero entries in the factorization is at most $N'=2n+\|E\|+N+\|E\|+2m$.
	
	
	To solve a system of linear equations $Mx = r$ with input vector $r$, we proceed with each matrix of the generalized factorization just as in Lemma~\ref{lem:usingGaussian}, either using back-substitution or permuting entries, using a total number of field operations equal to twice the number of non-zero entries and $\Oh(N')$ time.
\end{proof}

Given Lemma~\ref{lem:twsplitting}, from an $n\times m$ matrix $M$ and a tree decomposition of $G_M$ of width at most $b$, we can construct an $(n+N)\times(m+N)$ matrix $M_\treedecomp$ with a tree-partition decomposition of width at most $b$, in time $\Oh(N)$, for $N=\Oh((n+m)\cdot b)$.
Then, by Lemma~\ref{lem:stw-ordering} we can find a completion $H$ of $G_{M_\treedecomp}$ with at most $N'=(n+m+2N)\cdot b$ edges and a strong $H$-ordering of width at most $2b$.
Therefore, Gaussian elimination can be performed on $M_\treedecomp$ using $\Oh(N'\cdot b)=\Oh((n+m)\cdot b^3)$ field operations and time by Lemma~\ref{lem:gaussian}, yielding matrices with at most $2N'+n+m+2N=\Oh((n+m)\cdot b^2)$ non-zero entries by Lemma~\ref{lem:invariant}.
This allows us to retrieve the rank and determinant of $M_\treedecomp$, by Lemma~\ref{lem:usingGaussian}, and hence also of the original matrix $M$, by Lemma~\ref{lem:twsplitting}.
We can also retrieve a $PLUQ$-factorization of $M_\treedecomp$ and hence, by Lemma~\ref{lem:twLU}, a generalized $LU$-factorization of $M$ with $\Oh((n+m)\cdot b^2)$ non-zero entries, which allows us to solve a system of linear equations $Mx=r$ with given $r$ in $\Oh((n+m)\cdot b^2)$ additional field operations and time. 
This concludes the proof of Theorem~\ref{thm:tw-det}.

\restatetwdet*

%% file: figureSplitting.tex
\tikzstyle{v}=[circle,fill=black,draw=black!75,inner sep=0pt,minimum size=0.3em]

\begin{tikzpicture}
	\begin{scope}[shift={(-5,0)}]
	\node[blue] (u) at (0,0) {$r$};
	\node[red] (v) at (1,0) {$r'$};
	\node[orange] (w) at (-0.5,-0.8) {$c''$};
	\node[green!70!black] (x) at (0.5,1) {$c$};
	\node[yellow!70!black] (y) at (0.8,-1) {$c'$};
	\draw (u)--(y)--(v)--(x)--(u)--(w);
	\end{scope}
	
	\begin{scope}
	\draw[draw=none,fill=gray!50!white] (-1cm,0) rectangle (1cm,1.5);	
	\draw[draw=none,fill=gray!50!white] (-1cm,0) -- (-2.5,-1.5) -- (-0.5,-1.5) -- (1cm,0);			
	\draw[draw=none,fill=gray!50!white] (-1cm,0) -- (0.5,-1.5) -- (2.5,-1.5) -- (1cm,0);							
	
	\node[fill=gray!50!white,ellipse,minimum width=2cm,minimum height=1cm] at (0,1.5) {};
	\node[fill=white,ellipse,minimum width=1.7cm,minimum height=0.7cm] at (0,1.5) {$r, c, r'$};
	
	\node[fill=gray!50!white,ellipse,minimum width=2cm,minimum height=1cm] at (0,0) {};
	\node[fill=white,ellipse,minimum width=1.7cm,minimum height=0.7cm] at (0,0) {$r, r'$};
	
	\node[fill=gray!50!white,ellipse,minimum width=2cm,minimum height=1cm] at (-1.5,-1.5) {};
	\node[fill=white,ellipse,minimum width=1.8cm,minimum height=0.7cm] at (-1.5,-1.5) {$c'',r$};
	
	\node[fill=gray!50!white,ellipse,minimum width=2cm,minimum height=1cm] at (1.5,-1.5) {};
	\node[fill=white,ellipse,minimum width=1.8cm,minimum height=0.7cm] at (1.5,-1.5) {$r, c', r'$};
	
	
	\node[gray] at (0.7,1.5) {\Large\textbf{0}};	
	\node[gray] at (0.7,0) {\Large\textbf{1}};
	\node[gray] at (-0.8,-1.5) {\Large\textbf{2}};
	\node[gray] at (2.2,-1.5) {\Large\textbf{3}};		
	\end{scope}	
	
	\begin{scope}[shift={(6,0)},scale=1.1]
	\draw[draw=none,fill=gray!50!white] (-1cm,0) rectangle (1cm,1.5);	
	\draw[draw=none,fill=gray!50!white] (-1cm,0) -- (-2.5,-1.5) -- (-0.5,-1.5) -- (1cm,0);			
	\draw[draw=none,fill=gray!50!white] (-1cm,0) -- (0.5,-1.5) -- (2.5,-1.5) -- (1cm,0);							
	
	\node[fill=gray!50!white,ellipse,minimum width=2.2cm,minimum height=1cm] at (0,1.5) {};
	\node[fill=white,ellipse,minimum width=1.7cm,minimum height=0.7cm] at (0,1.5) {};
	
	\node[fill=gray!50!white,ellipse,minimum width=2.2cm,minimum height=1cm] at (0,0) {};
	\node[fill=white,ellipse,minimum width=1.7cm,minimum height=0.7cm] at (0,0) {};
	
	\node[fill=gray!50!white,ellipse,minimum width=2.2cm,minimum height=1cm] at (-1.5,-1.5) {};
	\node[fill=white,ellipse,minimum width=1.7cm,minimum height=0.7cm] at (-1.5,-1.5) {};
	
	\node[fill=gray!50!white,ellipse,minimum width=2.2cm,minimum height=1cm] at (1.5,-1.5) {};
	\node[fill=white,ellipse,minimum width=1.7cm,minimum height=0.7cm] at (1.5,-1.5) {};
	
	%
	
	\node[fill=white,ellipse,minimum width=1.5cm,minimum height=0.7cm] at (0,0.75) {};		
	\node[fill=white,ellipse,minimum width=1.5cm,minimum height=0.7cm] at (-0.8,-0.75) {};			
	\node[fill=white,ellipse,minimum width=1.5cm,minimum height=0.7cm] at (0.8,-0.75) {};			
	
	\node[v,blue,label={[label distance=0.3cm,blue!50!black]left:$r0$}] (r0) at (-0.6,1.5) {};
	\node[v,green!80!black,label={[label distance=0.3cm,green!50!black]above:$c0$}] (c0) at (-0,1.6) {};
	\node[v,red,label={[label distance=0.3cm,red!50!black]right:$r'0$}] (rp0) at (0.6,1.5) {};	
	
	\node[v,blue,label={[label distance=0.4cm,blue!50!black]left:\footnotesize$r01$}] (r01) at (-0.4,0.75) {};
	\node[v,red,label={[label distance=0.4cm,red!50!black]right:\footnotesize$r'01$}] (rp01) at (0.4,0.75) {};	
	
	\node[v,blue,label={[blue!50!black]left:$r1$}] (r1) at (-0.2,0) {};
	\node[v,red ,label={[ red!50!black]right:$r'1$}] (rp1) at (0.15,0) {};	
	
	\node[v,blue,label={[blue!50!black]left:\footnotesize$r12$}] (r12) at (-0.65,-0.75) {};	
	
	\node[v,blue,label={[blue!50!black]left:$r2$}] (r2) at (-1.8,-1.5) {};
	\node[v,orange,label={[orange!70!black]right:$c''2$}] (cpp2) at (-1.3,-1.6) {};
	
	\node[v,blue,label={[blue!50!black]left:\footnotesize$r13$}] (r13) at (0.6,-0.75) {};
	\node[v,red ,label={[red!50!black]right:\footnotesize$r'13$}] (rp13) at (0.9,-0.75) {};		
	%
	%
	
	\node[v,blue,label={[label distance=0.3cm,blue!50!black]left:$r3$}] (r4) at (1,-1.5) {};
	\node[v,yellow,label={[label distance=0.3cm,yellow!50!black]below:$c'3$}] (cp4) at (1.5,-1.6) {};	
	\node[v,red,label={[label distance=0.3cm,red!50!black]right:$r'3$}] (rp4) at (2,-1.5) {};
	
	\draw[thick] (r0) -- (c0) -- (rp0) (r2) -- (cpp2) (r4) -- (cp4) -- (rp4);
	
	\draw[blue,thick] (r0) -- (r01) -- (r1) -- (r12) --(r2) (r1)--(r13)--(r4);
	\draw[red,thick] (rp0) -- (rp01) -- (rp1) -- (rp13) -- (rp4);
	\end{scope}		
	
	\node at (-5,-4) { 
		\begin{tabular}{r | c c c}
		& $c$     &   $c'$  & $c''$\\\hline
		$r$ & $\star$ & $\star$ & $\star$ \\
		$r'$& $\star$ & $\star$ & 0
		\end{tabular}
	};
	
	\node at (-5,-6.5) { 
		\begin{tabular}{r | c c c c}
		& $c$     &   $c'$  & $c''$ & $r01$\\\hline
		$r0$ & $\star$ & 0 & 0 & 1\\
		$r1$ & 0 & $\star$ & $\star$ & -1\\	
		$r'$& $\star$ & $\star$ & 0 & 0
		\end{tabular}
	};
	
	\node at (5,-5.5) { 
		\begin{tabular}{r | c c c c c c c c}
		&$c0$&$c'3$&$c''2$&$r01$&$r12$&$r13$&$r'01$&$r'13$\\\hline
		$r0$ &$\star$& 0&  0   &  1  &  0  & 0 &0& 0\\
		$r1$ &0& 0&  0   &  -1  &  1  & 1 & 0 & 0\\		
		$r2$ &0& 0& $\star$ &  0  &  -1  & 0  &0& 0\\				
		$r3$ &0& $\star$&  0   &  0  &  0  &  -1 &0& 0\\				
		$r'0$ &$\star$& 0 & 0 & 0 & 0 & 0  & 1& 0\\
		$r'1$ & 0 & 0 & 0 & 0 & 0 & 0  & -1 & 1 \\
		$r'3$ & 0 & $\star$ & 0 & 0 & 0 & 0 & 0 & -1
		\end{tabular}
	};				
\end{tikzpicture}

%% file: matching.tex
\newcommand{\Vars}{\mathcal{E}}
\newcommand{\Z}{\mathbb{Z}}

\paragraph*{Computation model.} As we have already mentioned in the introduction, in this section we will consider arithmetic operations in a field $\F$ of size $\Oh(n^c)$ for some constant $c$. Although any such field would suffice, let us focus our attention on $\F=\F_p$ for some prime $p=\Oh(n^c)$. Then an element of $\F$ can be stored in a constant number of machine words of length $\log n$, and the following arithmetic operations in $\F_p$ can be easily implemented in constant time, assuming constant-time arithmetic on machine words of length $\log n$: addition, subtraction, multiplication, testing versus zero. The only operation that is not easily implementable in constant time is inversion, and hence also division. In a standard RAM machine, we need to apply the extended Euclid's algorithm, which takes $\Oh(\log n)$ time. However, in real-life applications, operations such as arithmetic in a field are likely to be heavily optimized, and hence we find it useful to separate this part of the running time.

In the algorithms in the sequel we state the consumed time resources in terms of {\em{time}} (standard operations performed by the algorithm) and {\em{field operations}}, each time meaning operations in some $\F_p$ for a polynomially bounded $p$. 

\paragraph*{Computing the size of a maximum matching.} To address the problem of computing the maximum matching in a graph, we need to recall some classic results on expressing this problem algebraically. Recall that a matrix $A$ is called {\em{skew-symmetric}} if $A=-A^T$. For a graph $G$ with vertex set $\{v_1,v_2,\ldots,v_n\}$, let $\Vars=\{x_{ij}\,\colon\, i<j,\, v_iv_j\in E(G)\}$ be a set of indeterminates associated with edges of $G$. With $G$ we can associate the {\em{Tutte matrix of $G$}}, denoted $\tilde{A}(G)$ and defined as follows. The matrix $\tilde{A}(G)=[a_{ij}]_{1\leq i,j\leq n}$ is an $n\times n$ matrix over the field $\Z(\Vars)$, i.e., the field of fractions of multivariate polynomials over $\Vars$ with integer coefficients. Its entries are defined as follows:
$$a_{ij}=\begin{cases} x_{ij} & \textrm{if $i<j$ and $v_iv_j\in E(G)$;}\\-x_{ji} & \textrm{if $i>j$ and $v_iv_j\in E(G)$;}\\ 0 & \textrm{otherwise.}\end{cases}$$
Clearly $\tilde{A}(G)$ is skew-symmetric. The following result is due to Tutte~\cite{Tutte} (for perfect matchings) and Lov\'asz~\cite{Lovasz79} (for maximum matchings).

\begin{theorem}[\cite{Tutte,Lovasz79}]\label{thm:tutte}
$\tilde{A}(G)$ is nonsingular if and only if $G$ has a perfect matching. Moreover, $\rk \tilde{A}(G)$ is equal to twice the maximum size of a matching in $G$.
\end{theorem}

From Theorem~\ref{thm:tutte} we can derive our first result on finding the cardinality of a maximum matching, that is, Theorem~\ref{thm:matching-size}, which we recall below.

\restatesize*
\begin{proof}
Arbitrarily enumerate $V(G)$ as $\{v_1,v_2,\ldots,v_n\}$. Let $\tilde{A}(G)$ be the Tutte matrix of $G$. Let $p$ be a prime with $n^{c+1}\leq p< 2\cdot n^{c+1}$, and let $\F=\F_{p}$ ($p$ can be found in polylogarithmic expected time by iteratively sampling a number and checking whether it is prime using the AKS algorithm). 

Construct matrix $A(G)$ from $\tilde{A}(G)$ by substituting each indeterminate from $\Vars$ with a value chosen uniformly and independently at random from $\F$. Since the determinant of the largest nonsingular square submatrix of $\tilde{A}(G)$ is a polynomial over $\Vars$ of degree at most $n$, from Schwarz-Zippel lemma it follows that this submatrix remains nonsingular in $A(G)$ with probability at least $1-\frac{n}{p}\geq 1-\frac{1}{n^c}$. Moreover, for any square submatrix of $A(G)$ that is nonsingular, the corresponding submatrix of $\tilde{A}(G)$ is nonsingular as well. Hence, with probability at least $1-\frac{1}{n^c}$ we have that matrix $A(G)$ has the same rank as $\tilde{A}(G)$, and otherwise the rank of $A(G)$ is smaller than that of $\tilde{A}(G)$.

Let $H$ be the bipartite graph $G_{A(G)}$ associated with matrix $A(G)$. Then $H$ has $2n$ vertices: for every vertex $u$ of $G$, $H$ contains a {\em{row-copy}} and a {\em{column-copy}} of $u$. Based on a decomposition of $G$ of width $k$, it is easy to construct a tree decomposition of $H$ of width at most $2k+1$ as follows: the decomposition has the same tree, and in each bag we replace each vertex of $G$ by both its copies in $H$. The construction of $H$ and its tree decomposition takes time $\Oh(kn)$.

We now use the algorithm of Theorem~\ref{thm:tw-det} to compute the rank of $A(G)$; this uses $\Oh(k^3\cdot n)$ time and field operations. Supposing that $A(G)$ indeed has the same rank as $\tilde{A}(G)$ (which happens with probability at least $1-\frac{1}{n^c}$), we have by Theorem~\ref{thm:tutte} that the size of a maximum matching in $G$ is equal to the half of this rank, so we report this value. In the case when $\rk A(G)<\rk \tilde{A}(G)$, which happens with probability at most $\frac{1}{n^c}$, the algorithm will report a value smaller than the maximum size of a matching in $G$.
\end{proof}

\paragraph*{Reconstructing a maximum matching.} Theorem~\ref{thm:matching-size} gives only the maximum size of a matching, but not the matching itself. To recover the maximum matching we need some extra work. 

First, we need to perform an analogue of the splitting operation from Section~\ref{sec:gaussian} in order to reduce the case of tree decompositions to tree-partition decompositions. The reason is that Theorem~\ref{thm:tw-det} does not give us a maximum nonsingular submatrix of the Tutte matrix of the graph, which we need in order to reduce finding a maximum matching to finding a perfect matching in a subgraph. 

\newcommand{\off}{\Lambda}

\begin{lemma}\label{lem:reduction-tw-tpw}
There exists an algorithm that given a graph $G$ together with its clean tree decomposition $(\Tt,\{B_x\}_{x\in V(\Tt)})$ of width at most $k$, constructs another graph $G'$ together with its tree-partition decomposition $(\Tt',\{C_x\}_{x\in V(\Tt')})$ of width at most $k$, such that the following holds:
\begin{enumerate}[(i)]
\item\label{pr:time2} The algorithm runs in time $\Oh(kn)$;
\item\label{pr:sz} $|V(G')|\leq 2kn$;
\item\label{pr:GtoGp} Given a matching $M$ in $G$, one can construct in $\Oh(kn)$ time a matching $M'$ in $G'$ with $|M'|=|M|+\off/2$, where $\off=|V(G')|-|V(G)|$;
\item\label{pr:GptoG} Given a matching $M'$ in $G'$, one can construct in $\Oh(kn)$ time a matching $M$ in $G$ with $|M|\geq |M'|-\off/2$.
\end{enumerate}
\end{lemma}
\begin{proof}
The vertex set of $G'$ consists of the following vertices:
\begin{itemize}
\item For every vertex $u\in V(G)$ and every node $t\in V(\Tt)$ with $u\in B_t$, create a vertex $(u,t)$;
\item For every vertex $u\in V(G)$ and every pair of adjacent nodes $t,t'\in V(\Tt)$ with $u\in B_t\cap B_{t'}$, create a vertex $(u,tt')$.
\end{itemize}
The edge set of $G'$ is defined as follows:
\begin{itemize}
\item For every vertex $u\in V(G)$ and every pair of adjacent nodes $t,t'\in V(\Tt)$ with $u\in B_t\cap B_{t'}$, add edges $(u,tt')(u,t)$ and $(u,tt')(u,t')$.
\item For every edge $uv\in E(G)$ and every node $t\in V(\Tt)$ with $\{u,v\}\subseteq B_t$, create an edge $(u,t)(v,t)$. 
\end{itemize}
This finishes the description of $G'$. Construct $\Tt'$ from $\Tt$ by subdividing every edge of $\Tt$ once, i.e., $\Tt'$ is the $1$-subdivision of $\Tt$. Let the subdivision node used to subdivide edge $tt'$ be also called $tt'$. Place all the vertices of $G'$ of the form $(u,t)$ in a bag $C_t$, for each $t\in V(\Tt)$, and place all the vertices of $G'$ of the form $(u,tt')$ in a bag $C_{tt'}$, for each $tt'\in E(\Tt)$. Then it readily follows that $(\Tt',\{C_x\}_{x\in V(\Tt')})$ is a tree-partition decomposition of $G'$ of width at most $k$. Moreover, since $\Tt$ was clean, it has at most most $n$ nodes
, and hence also at most $n-1$ edges. Hence $G'$ has at most $2kn$ vertices, and~\eqref{pr:sz} is satisfied. It is straightforward to implement the construction above in time $\Oh(kn)$, and hence~\eqref{pr:time2} also follows. We are left with showing how matchings in $G$ and $G'$ can be transformed one to the other.

Before we proceed, let us introduce some notation. Let $W^V$ be the set of all the vertices of $G'$ of the form $(u,t)$ for some $t\in V(\Tt)$, and let $W^E$ be the set of all the vertices of $G'$ of the form $(u,tt')$ for some $tt'\in E(\Tt)$. Then $(W^V,W^E)$ is a partition of $V(G')$. For some $u\in V(G)$, let $T_u$ be the subgraph of $G'$ induced by all the vertices with $u$ on the first coordinate, i.e., of the form $(u,t)$ or $(u,tt')$. It follows that $T_u$ is a tree. Moreover, if we denote $W^V_u=W^V\cap V(T_u)$ and $W^E_u=W^E\cap V(T_u)$, then $(W^V_u,W^E_u)$ is a bipartition of $T_u$, all the vertices of $W^E_u$ have degrees $2$ in $T_u$ and no neighbors outside $T_u$, and $|W^V_u|=|W^E_u|+1$. Then, we have that 
$$|V(G')|=\sum_{u\in V(G)} |W^V_u|+|W^E_u|=\sum_{u\in V(G)} (2|W^E_u|+1)=|V(G)|+2\sum_{u\in V(G)} |W^E_u|,$$
so $\off/2=\sum_{u\in V(G)} |W^E_u|$.

\begin{claim}\label{cl0}
For each $u\in V(G)$ and each $w\in W^V_u$, there is a matching $M_{u,w}$ in $T_u$ that matches all the vertices of $T_u$ apart from $w$. Moreover, $M_{u,w}$ can be constructed in time $\Oh(|V(T_u)|)$.
\end{claim}
\begin{proof}
Root $T_u$ in $w$. This imposes a parent-child relation on the vertices of $T_u$, and in particular each vertex of $W^E_u$ has exactly one child, which of course belongs to $W^V_u$. Construct $M_{u,w}$ by matching each vertex of $W^E_u$ with its only child. Then only the root $w$ is left unmatched in $T_u$.
\cqed\end{proof}

Now, the first direction is apparent.

\begin{claim}\label{cl1}
Condition~\eqref{pr:GtoGp} holds.
\end{claim}
\begin{proof}
Let $M$ be a given matching in $G$. Construct $M'$ as follows: For every $uv\in M$, pick an arbitrary node $t$ of $\Tt$ with $\{u,v\}\subseteq B_t$, and add $(u,t)(v,t)$ to $M'$. After this, from each subtree $T_u$, for $u\in V(G)$, at most one vertex is matched so far, and it must belong to $W^V_u$. Let $w_u\in V(T_u)$ be this matched vertex in $T_u$; in case no vertex of $T_u$ is matched so far, we take $w_u$ to be an arbitrary vertex of $W^V_u$. For each $u\in V(G)$, we add to $M'$ the matching $M_{u,w_u}$, which has cardinality $|W^E_u|$; this concludes the construction of $M'$. It is clear that $M'$ is a matching in $G'$, and moreover 
$$|M'|=|M|+\sum_{u\in V(G)} |W^E_u| = |M|+\off/2.$$
It is easy to see that a straightforward construction of $M'$ takes time $\Oh(kn)$.
\cqed\end{proof}

We now proceed to the proof of the second direction.

\begin{claim}\label{cl2}
Condition~\eqref{pr:GptoG} holds.
\end{claim}
\begin{proof}
Let $M'$ be a given matching in $G'$. We first transform it to a matching $M''$ in $G'$
such that $|M''|\geq |M'|$ and for each $u\in V(G)$ at most one vertex from $T_u$ is matched in $M'$ with a vertex outside $T_u$.
If for some vertex $u$ this is not true, we arbitrarily select one edge $ww'$ of those in the matching with exactly one endpoint, say $w$, in $T_u$ (and hence in $W^V_u$).
We then remove all edges incident to $T_u$ from the matching except $ww'$ and add $M_{u,w}\subseteq E(T_u)$ from Claim~\ref{cl0} to the matching.
We removed at most $|W^V_u|-1$ edges, as every edge of a matching incident to $T_u$ must contain a vertex of $W^V_u$.
However, we added $|W^E_u|=|W^V_u|-1$ edges, hence the matching could only increase in size.

Construct the matching $M$ in $G$ as follows: Inspect all the edges of $M''$, and for each edge of the form $(u,t)(v,t)\in M''$ for some $u,v\in V(G)$, add the edge $uv$ to $M$. The 
construction of $M''$ ensures that $M$ is indeed a matching in $G$, and no edge of $G$ is added twice to $M$. Moreover, for each $u\in V(G)$, we have that $M''\cap E(T_u)$ is a matching in $T_u$, so in particular $|M''\cap E(T_u)|\leq |W^E_u|$. Hence, from the construction we infer that
$$|M'|\leq |M''|=|M|+\sum_{u\in V(G)} |M''\cap E(T_u)|\leq |M|+\sum_{u\in V(G)} |W^E_u|=|M|+\off/2.$$
A straightforward implementation of the construction of $M$ runs in time $\Oh(kn)$.
\cqed\end{proof}

Claims~\ref{cl1} and~\ref{cl2} conclude the proof.
\end{proof}

Next, we need to recall how finding a maximum matching can be reduced to finding a perfect matching using a theorem of Frobenius and the ability to efficiently compute a largest nonsingular submatrix; this strategy was used e.g. by Mucha and Sankowski in~\cite{MuchaS06}.

\begin{theorem}[Frobenius Theorem]
Suppose $A$ is an $n\times n$ skew-symmetric matrix, and suppose $X,Y\subseteq [n]$ are such that $|X|=|Y|=\rk A$. Then
$$\det [A]_{X,X}\cdot \det [A]_{Y,Y} = (-1)^{|X|}\cdot \left(\det [A]_{X,Y}\right)^2.$$
\end{theorem}

\begin{corollary}\label{cor:base-matching}
Let $G$ be a graph with vertex set $\{v_1,v_2,\ldots,v_n\}$, and suppose $[\tilde{A}(G)]_{X,Y}$ is a maximal nonsingular submatrix of $\tilde{A}(G)$.
Then $\overline{X}=\{v_i\colon i\in X\}$ is a subset of $V(G)$ of maximum size for which $G[\overline{X}]$ contains a perfect matching.
\end{corollary}
\begin{proof}
By the assumption about $X$ we have that $|X|=\rk \tilde{A}(G)$. Since
$[\tilde{A}(G)]_{X,Y}$ is nonsingular,
by the Frobenius Theorem both $[\tilde{A}(G)]_{X,X}$ and $[\tilde{A}(G)]_{Y,Y}$ are nonsingular as well. Observe that $[\tilde{A}(G)]_{X,X}=\tilde{A}(G[\overline{X}])$, so by Theorem~\ref{thm:tutte} we infer that $G[\overline{X}]$ contains a perfect matching. Moreover, $\overline{X}$ has to be the largest possible set with this property, because $|\overline{X}|=\rk \tilde{A}(G)$.
\end{proof}

Now we are ready to prove the analogue of Theorem~\ref{thm:matching-reconstruct} for tree-partition width. This proof contains the main novel idea of this section. Namely, it is known if a perfect matching is present in the graph, then from the inverse of the Tutte matrix one can derive the information about which edges are contained in some perfect matching. Finding one column of the inverse amounts to solving one system of linear equations, so in time roughly $\Oh(k^2 n)$ we are able to find a suitable matching edge for any vertex of the graph. Having found such an edge, both the vertex and its match can be removed; we call this operation {\em{ousting}}. However, ousting iteratively on all the vertices would result in a quadratic dependence on $n$. Instead, we apply a Divide\&Conquer scheme: we first oust the vertices of a balanced bag of the give tree-partition decomposition, and then recurse on the connected components of the remaining graph. This results in $\Oh(\log n)$ levels of recursion, and the total work used on each level is $\Oh(k^3\cdot n)$. We remark that in Section~\ref{sec:max-flow} we apply a similar approach to the maximum flow problem.

\begin{lemma}\label{lem:matching-tpw}
There exists an algorithm that, given a graph $G$ together with its tree-partition decomposition of width at most $k$, uses $\Oh(k^3\cdot n\log n)$ time and field operations and computes a maximum matching in $G$. The algorithm is randomized with one-sided error: it is correct with probability at least $1-\frac{1}{n^c}$ for an arbitrarily chosen constant $c$, and in the case of an error it reports a failure or a sub-optimal matching.
\end{lemma}
\begin{proof}
Let $\{v_1,v_2,\ldots,v_n\}$ be an arbitrary enumeration of the vertices of $G$ and let $\tilde{A}(G)$ be the Tutte matrix of $G$. For the entire proof we fix a field $\F=\F_p$ for some prime $p$ with $n^{c+5}\leq p< 2\cdot n^{c+5}$. Construct $A(G)$ from $\tilde{A}(G)$ by substituting each indeterminate from $\Vars$ with a value chosen uniformly and independently at random from $\F$. Since the determinant of the largest nonsingular submatrix of $\tilde{A}(G)$ is a polynomial over $\Vars$ of degree at most $n$, from the Schwarz-Zippel lemma we obtain that this submatrix remains nonsingular in $A(G)$ with probability at least $1-\frac{n}{n^{c+5}}=1-\frac{1}{n^{c+4}}$; of course, in this case $A(G)$ has the same rank as $\tilde{A}(G)$.

Similarly as in the proof of Theorem~\ref{thm:matching-size}, let $H=G_{A(G)}$ be the bipartite graph associated with $A(G)$. Then, a tree-partition decomposition of $H$ of width at most $2k$ can be constructed from the given tree-partition decomposition of $G$ of width at most $k$ by substituting every vertex with its row-copy and column-copy in $H$ in the corresponding bag. Therefore, we can apply the algorithm of Theorem~\ref{thm:pw-det} to $A(G)$ with this decomposition of $H$, and retrieve a largest nonsingular submatrix of $A(G)$ using $\Oh(k^2\cdot n)$ time and field operations. Suppose that this submatrix is $[A(G)]_{X,Y}$ for some $X,Y\subseteq [n]$. Then $|X|=|Y|=\rk A(G)$ and of course $[\tilde{A}(G)]_{X,Y}$ is nonsingular as well. Recall that with probability at least $1-\frac{1}{n^{c+4}}$ we have $\rk A(G)=\rk\tilde{A}(G)$, so if this holds, then $[\tilde{A}(G)]_{X,Y}$ is also the largest nonsingular submatrix of $\tilde{A}(G)$. Then the rows of $\tilde{A}(G)$ with indices of $X$ form a base of the subspace of $\Z(\Vars)^n$ spanned by all the rows of $\tilde{A}(G)$. By Corollary~\ref{cor:base-matching} we infer that $\overline{X}=\{v_i\colon i\in X\}$ is the maximum size subset of $V(G)$ for which $G[\overline{X}]$ contains a perfect matching.

We now constrain our attention to the graph $G[\overline{X}]$, and our goal is to find a perfect matching there. Observe that a tree-partition decomposition of $G[\overline{X}]$ of width at most $k$ can be obtained by taking the input tree-partition decomposition of $G$, removing all the vertices of $V(G)\setminus \overline{X}$ from all the bags, and deleting bags that became empty. Hence, in this manner we effectively reduced finding a maximum matching to finding a perfect matching, and from now on we can assume w.l.o.g. that the input graph $G$ has a perfect matching.

If $G$ has a perfect matching, then from Theorem~\ref{thm:tutte} we infer that $\tilde{A}(G)$ is nonsingular. Again, construct $A(G)$ from $\tilde{A}(G)$ by substituting each indeterminate from $\Vars$ with a value chosen uniformly and independently at random from $\F$. Since $\det \tilde{A}(G)$ is a polynomial of degree at most $n$ over $\Vars$, from Schwarz-Zippel lemma we infer that with probability at least $1-\frac{n}{n^{c+5}}=1-\frac{1}{n^{c+4}}$ matrix $A(G)$ remains nonsingular. Similarly as before, a suitable tree-partition decomposition of the bipartite graph $H=G_{A(G)}$ can be constructed from the input tree-partition decomposition of $G$, and hence we can apply Theorem~\ref{thm:pw-det} to $A(G)$. In particular, we can check using $\Oh(k^2\cdot n)$ time and field operations whether $A(G)$ is indeed nonsingular, and otherwise we abort the computations and report failure.

Call an edge $e\in E(G)$ {\em{allowed}} if there is some perfect matching $M$ in $G$ that contains $e$. The following result of Rabin and Vazirani~\cite{RabinV89} shows how it can be retrieved from the inverse of $\tilde{A}(G)$ whether an edge is allowed; it is also the cornerstone of the approach of Mucha and Sankowski~\cite{MuchaS06,MuchaS04}.

\begin{claim}[\cite{RabinV89}]\label{cl:allowed}
Edge $v_iv_j$ is allowed if and only if the entry $(\tilde{A}(G)^{-1})[j,i]$ is non-zero.
\end{claim}

This very easily leads to an algorithm that identifies some allowed edge incident on a vertex.

\begin{claim}\label{cl:allowedFind}
For any vertex $u\in V(G)$, one can find an allowed edge incident on $u$ using $\Oh(k^2 n)$ time and field operations.
\end{claim}
\begin{proof}
Suppose $u=v_i$ for some $i\in [n]$. We recover the $i$-th column $c_i$ of $A(G)^{-1}$ by solving the system of equations $A(G)c_i = e_i$, where $e_i$ is the unit vector with $1_\F$ on the $i$-th coordinate, and zeros on the other coordinates. Using Theorem~\ref{thm:pw-det} this takes time $\Oh(k^2 n)$, including the time needed for Gaussian elimination. Now observe that $c_i$ must have a non-zero entry on some coordinate, say the $j$-th, which corresponds to an edge $v_iv_j\in E(G)$. Indeed, the $i$-th row $r_i$ of $A(G)$ has non-zero entries only on the coordinates that correspond to neighbors of $v_i$ and we have that $r_ic_i=1_\F$, so it cannot happen that all the coordinates of $c_i$ corresponding to the neighbors of $v_i$ have zero values.

This means that $(A(G)^{-1})[j,i]\neq 0$, which implies that $(\tilde{A}(G)^{-1})[j,i]\neq 0$. Since $v_iv_j\in E(G)$, by Claim~\ref{cl:allowed} this means that $v_iv_j$ is allowed.
\cqed\end{proof}

Let us introduce the operation of {\em{ousting}} a vertex $u$, defined as follows: Apply the algorithm of Claim~\ref{cl:allowedFind} to find an allowed edge $uv$ incident on $u$, add $uv$ to a constructed matching $M$, and remove both $u$ and $v$ from the graph. Clearly, the resulting graph $G'$ still has a perfect matching due to $uv$ being allowed, so we can proceed on $G'$. Note that we can compute a suitable tree decomposition of $G'$ by deleting $u$ and $v$ from the corresponding bags of the current tree-partition decomposition, and removing bags that became empty. By Claim~\ref{cl:allowedFind}, any vertex can be ousted within $\Oh(k^2\cdot n)$ time and field operations.

We can now describe the whole algorithm. Let $(\Tt,\{B_x\}_{x\in V(\Tt)})$ be the given tree-partition decomposition of $G$, and let $q=|V(\Tt)|$; as each bag can be assumed to be non-empty, we have $q\leq n$. Define a uniform measure $\meas$ on $V(\Tt)$: $\meas(x)=1$ for each $x\in V(\Tt)$. Using the algorithm of Lemma~\ref{lem:balanced}, find in time $\Oh(q)$ a node $x\in V(\Tt)$ such that each connected component of $\Tt-x$ has at most $q/2$ nodes. Iteratively apply the ousting procedure to all the vertices of $B_x$; each application boils down to solving a system of equations using $\Oh(k^2\cdot n)$ time and field operations, so all the ousting procedures in total use $\Oh(k^3\cdot n)$ time and field operations. Note that in consecutive oustings we work on a graph with fewer and fewer vertices, so after each ousting we again resample the Tutte matrix of the current graph, and perform the Gaussian elimination again.

Let $G'$ be the graph after applying all the ousting procedures. For every subtree $\Cc\in \cc(\Tt-x)$, let $G_\Cc$ be the subgraph of $G'$ induced by the vertices of $G'$ that are placed in the bags of $\Cc$. Then graphs $G_\Cc$ for $\Cc\in \cc(\Tt-x)$ are pairwise disjoint and non-adjacent, and their union is $G'$. Since $G'$ has a perfect matching (by the properties of ousting), so does each $G_\Cc$. Observe that a tree-partition decomposition $\Tt_\Cc$ of $G_\Cc$ of width at most $k$ can be obtained by inspecting $\Cc$, removing from each bag all the vertices not belonging to $G_\Cc$, and removing all the bags that became empty in this manner. Hence, to retrieve a perfect matching of $G$ it suffices to apply the algorithm recursively on all the instances $(G_\Cc,\Tt_\Cc)$, and take the union of the retrieved matchings together with the edges gathered during ousting.

We first argue that the whole algorithm runs in time $\Oh(k^3\cdot n\log n)$. Let $\Rt$ denote the recursion tree of the algorithm --- a rooted tree with nodes labelled by instances solved in recursive subcalls, where the root corresponds to the original instance $(G,\Tt)$ and the children of each node correspond to subcalls invoked when solving the instance associated with this node. Note that the leaves of $\Rt$ correspond to instances where there is only one bag. Since the number of bags is halved in each recursive subcall, and at the beginning we have $q\leq n$ bags, we immediately obtain the following:
\begin{claim}\label{cl:depthmatching}
The height of $\Rt$ is at most $\log_2 n$.
\end{claim}

We partition the work used by the algorithm among the nodes of $\Rt$. Each node $a$ is charged by the work needed to (i) find the balanced node $x$, (ii) oust the vertices of $B_x$, (iii) construct the subcalls $(G_\Cc,\Tt_\Cc)$ for $\Cc\in \cc(\Tt-x)$, (iv) gather the retrieved matchings and the ousted edges into a perfect matching of the graph. From the description above it follows that if in the instance associated with a node $a$ the graph has $n_a$ vertices, then the total work associated with $a$ is $\Oh(k^3\cdot n_a)$.
Since subinstances solved at each level of the recursion correspond to subgraphs of $G$ that are pairwise disjoint, and since there are at most $\log_2 n$ levels,
the total work used by the algorithm is $\Oh(k^3\cdot n\log n)$.
%

We now estimate the error probability of the algorithm. Since on each level $L_i$ all the graphs associated with the subcalls are disjoint, we infer that the total number of subcalls is $\Oh(n\log n)$. In each subcall we apply ousting at most $k$ times, and each time we resample new elements of $\F$ to the Tutte matrix, which can lead to aborting the computations with probability at most $\frac{1}{n^{c+4}}$ in case the sampling led to making the matrix singular. By union bound, no such event will occur with probability at least $1-\frac{k\cdot n\log n}{n^{c+4}}\geq 1-\frac{1}{n^{c+1}}$. Together with the error probability of at most $\frac{1}{n^{c+4}}$ when initially reducing finding a maximum matching to finding a perfect matching, this proves that the error probability is at most $\frac{1}{n^{c+1}}+\frac{1}{n^{c+4}}\leq \frac{1}{n^c}$.
\end{proof}

We can now put all the pieces together and prove Theorem~\ref{thm:matching-reconstruct}, which we now restate for the reader's convenience.

\restatereconst*
\begin{proof}
Apply the algorithm of Lemma~\ref{lem:reduction-tw-tpw} to construct a graph $G'$ together with a tree-partition decomposition $\Tt'$ of $G'$ of width at most $k$; this takes time $\Oh(kn)$. Apply the algorithm of Lemma~\ref{lem:matching-tpw} to construct a maximum matching $M'$ in $G'$; this uses $\Oh(k^4\cdot n\log n)$ time and field operations, because $G'$ has $\Oh(kn)$ vertices. Using Lemma~\ref{lem:reduction-tw-tpw}\eqref{pr:GptoG}, construct in time $\Oh(kn)$ a matching $M$ in $G$ of size at least $|M'|-\off/2$, where $\off=|V(G')|-|V(G)|$. Observe now that $M$ has to be a maximum matching in $G$, because otherwise, by Lemma~\ref{lem:reduction-tw-tpw}\eqref{pr:GtoGp}, there would be a matching in $G'$ of size larger than $|M|+\off/2\geq |M'|$, which is a contradiction with the maximality of $M'$. In case the algorithm of Lemma~\ref{lem:matching-tpw} reported failure or returned a suboptimal matching of $G'$, which happens with probability at most $\frac{1}{n^c}$, the same outcome is given by the constructed procedure.
\end{proof}

Finally, we remark that actually the number of performed inversions in the field in the algorithms of Theorems~\ref{thm:matching-size} and~\ref{thm:matching-reconstruct} is $\Oh(kn)$ and $\Oh(k^2n)$, respectively. Hence, if inversion is assumed to cost $\Oh(\log n)$ time, while all the other operations take constant time, the running times of Theorems~\ref{thm:matching-size} and~\ref{thm:matching-reconstruct} can be stated as $\Oh(k^3\cdot n +k\cdot n\log n)$ and $\Oh(k^4\cdot n\log n +k^2\cdot n\log^2 n)$, respectively.

%% file: max-flow.tex
\newcommand{\conn}{\kappa}
\newcommand{\Iff}{\mathfrak{I}}
\newcommand{\apx}{\textrm{apx}}
\newcommand{\mexp}{\textrm{exp}}

In this section we prove Theorem~\ref{thm:max-flow}. Let us remind that our definition of $(s,t)$-vertex flows corresponds to a family of vertex-disjoint paths from $s$ to $t$; that is, we work with directed graphs with unit capacities on vertices.

\restatemaxflow*
\begin{proof}
Assume without loss of generality that the given tree decomposition $(\Tt,\{B_x\}_{x\in V(\Tt)})$ of $G$ is clean (otherwise clean it in linear time) and let $q \leq n$ be the number of its nodes.
For a graph $H$ containing $s$ and $t$, by $\conn_H(s,t)$ we will denote the maximum size of a flow from $s$ to $t$ in $H$. 

Let $L$ be the set of all the nodes $x\in V(\Tt)$ for which $\{s,t\}\subseteq B_x$; note that $L$ can be constructed in time $\Oh(kq)$. By the properties of a tree decomposition, $L$ is either empty or it induces a connected subtree of $\Tt$. Let $\ell=|L|$. We will design an algorithm with running time $\Oh(k^2\cdot n\log (\ell+2))$, which clearly suffices, due to $\ell\leq q\leq n$.

First, we need to introduce several definitions and prove some auxiliary claims.

\begin{claim}\label{cl:emptyL}
If $\ell=0$, then $\kappa_G(s,t)\leq k$.
\end{claim}
\begin{proof}
The assumption that $\ell=0$ means that there is no bag where $s$ and $t$ appear simultaneously. Since $\Tt[s]$ and $\Tt[t]$ are disjoint connected subtrees of $\Tt$, it follows that there is an edge $xy$ of $\Tt$ whose removal disconnects $\Tt$ into subtrees $\Tt_x$ (containing $x$) and $\Tt_y$ (containing $y$) such that $\Tt[s]\subseteq \Tt_x$ and $\Tt[t]\subseteq \Tt_y$. By the properties of a tree decomposition, we have that $B_x\cap B_y$ is an $(s,t)$-vertex separator. Since $B_x\neq B_y$ due to $\Tt$ being clean, we have that $|B_x\cap B_y|\leq k$. Hence in particular $\conn(s,t)\leq |B_x\cap B_y|\leq k$.
\cqed\end{proof}

Suppose now that $\ell>0$ and let us fix some arbitrary node $x$ belonging to $L$. Let us examine a component $\Cc\in \cc(\Tt-x)$, which is a subtree of $\Tt$. Let $W_\Cc$ consist of all the vertices $u\in V(G)$ for which $\Tt[u]\subseteq \Cc$, and let $y_\Cc$ be the (unique) neighbor node of $x$ contained in $\Cc$. We call subtree $\Cc$ {\em{important}} if $L\cap V(\Cc)\neq \emptyset$, i.e., $\Cc$ has at least one node that contains both $s$ and $t$; since $L$ induces a connected subtree of $\Tt$, this is equivalent to $y_\Cc\in L$. Otherwise $\Cc$ is called {\em{unimportant}}. A path $P$ from $s$ to $t$ is called {\em{expensive}} if $V(P)\cap (B_x\setminus \{s,t\})\neq \emptyset$, i.e., $P$ traverses at least one vertex of $B_x$ other than $s$ and $t$; otherwise, $P$ is called {\em{cheap}}.

\begin{claim}\label{cl:expensive}
For any $x\in L$ and any $(s,t)$-vertex flow $\Ff$, the number of expensive paths in $\Ff$ is at most $k-1$.
\end{claim}
\begin{proof}
The paths from $\Ff$ are internally vertex-disjoint and each of them traverses a vertex of $B_x\setminus \{s,t\}$. Since $|B_x\setminus \{s,t\}|\leq k-1$, the claim follows.
\cqed\end{proof}

\begin{claim}\label{cl:cheap}
For any $x\in L$ and any cheap $(s,t)$-path $P$, there exists a subtree $\Cc_P\in \cc(\Tt-x)$ such that $V(P)\subseteq W_{\Cc_P}\cup \{s,t\}$. Moreover, $\Cc_P$ is important. 
\end{claim}
\begin{proof}
By the properties of a tree decomposition, we have that the vertex set of each connected component of $G-B_x$ is contained in $W_\Cc$ for some subtree $\Cc\in \cc(\Tt-x)$. Since $P$ is cheap, we have that $V(P)\cap B_x=\{s,t\}$, and hence all the internal vertices of $P$ have to belong to the same connected component $H$ of $G-B_x$. Therefore, we can take $\Cc_P$ to be the connected component of $\Tt-x$ for which $V(H)\subseteq W_{\Cc_P}$. It remains to show that $\Cc_P$ is important.

For the sake of contradiction, suppose $\Cc_P$ is not important. Let $y=y_{\Cc_P}$; then $\{s,t\}\nsubseteq B_{y}$. Without loss of generality suppose $t\notin B_y$, as the second case is symmetric. Since $t\in B_x$, this means that no bag of $\Cc_P$ contains $t$. Consequently, there is no edge from any vertex of $W_{\Cc_P}$ to $t$ in $G$. However, all the internal vertices of $P$ are contained in $W_{\Cc_P}$, which is a contradiction.
\cqed\end{proof}

We can now describe the whole algorithm. First, the algorithm computes $L$ in time $\Oh(kq)$. If $L=\emptyset$, then by Claim~\ref{cl:emptyL} we have that $\kappa_G(s,t)\leq k$. Hence, by starting with an empty flow and applying at most $k$ times flow augmentation (Lemma~\ref{thm:augm}) we obtain both a maximum $(s,t)$-vertex flow and a minimum $(s,t)$-vertex cut in time $\Oh(k\cdot (n+m))=\Oh(k^2\cdot n)$. Therefore, suppose that $L\neq \emptyset$.

The crux of our approach is to take $x$ to be a node that splits $L$ in a balanced way. For this, Lemma~\ref{lem:bal-bag} will be useful. Define measure $\meas_L$ on $V(\Tt)$ as follows: for $x\in V(\Tt)$, let $\meas_L(x)=1$ if $x\in L$ and $\meas_L(x)=0$ otherwise. Since $L\neq \emptyset$, $\meas_L$ is indeed a measure. Run the algorithm of Lemma~\ref{lem:bal-bag} on $\Tt$ with measure $\meas_L$; let $x$ be the obtained balanced node. Then for each $\Cc\in \cc(\Tt-x)$ we have that $|L\cap V(\Cc)|=\meas_L(V(\Cc))\leq \meas_L(V(\Tt))/2=|L|/2$. In the following, the notions of expensive and cheap components refer to components of $\Tt-x$.

For each $\Cc\in \cc(\Tt-x)$, decide whether $\Cc$ is important by checking whether $\{s,t\}\subseteq B_{y_\Cc}$; this takes at most $\Oh(kq)$ time in total. Let $\Iff\subseteq \cc(\Tt-x)$ be the family of all the important subtrees of $\Tt-x$. For each $\Cc\in \Iff$, construct an instance $(G_\Cc,s,t,\Tt_\Cc)$ of the maximum vertex flow problem as follows: take $G_\Cc=G[W_\Cc\cup \{s,t\}]$, construct its tree decomposition $\Tt_\Cc$ from $\Cc$ by removing all the vertices not contained in $W_\Cc\cup \{s,t\}$ from all the bags of $\Cc$, and make it clean in linear time. 
Clearly, $\Tt_\Cc$ constructed in this manner is a tree decomposition of $G_\Cc$ of width at most $k$, and has at most half as many bags containing both $s$ and $t$ as $\Tt$. Moreover, a straightforward implementation of the construction of all the instances $(G_\Cc,s,t,\Tt_\Cc)$ for $\Cc\in \Iff$ runs in total time $\Oh(kn)$, due to $q\leq n$.

Now, apply the algorithm recursively to all the constructed instances $(G_\Cc,s,t,\Tt_\Cc)$, for all $\Cc\in \Iff$. Each application returns a maximum $(s,t)$-vertex flow $\Ff_\Cc$ in $G_\Cc$. Observe that all the internal vertices of all the paths of $\Ff_\Cc$ are contained in $W_\Cc$, and all the sets $W_\Cc$ for $\Cc\in \Iff$ are pairwise disjoint. Therefore, taking $\Ff^{\apx}=\bigcup_{\Cc\in \Iff} \Ff_\Cc$ defines an $(s,t)$-vertex flow, because the paths of $\Ff^{\apx}$ are internally vertex-disjoint.

Now comes the crucial observation: $\Ff_\apx$ is not far from the maximum flow.

\begin{claim}\label{cl:almost-done}
$|\Ff^\apx|\geq \conn_G(s,t)-(k-1)$.
\end{claim}
\begin{proof}
Let $\Ff^\circ$ be a maximum $(s,t)$-vertex flow. Partition $\Ff^\circ$ as $\Ff^\circ_\mexp\uplus \biguplus_{\Cc\in \Iff} \Ff^\circ_\Cc$, where $\Ff^\circ_\mexp$ is the set of the expensive paths from $\Ff^\circ$, whereas for $\Cc\in \Iff$ we define $\Ff^\circ_\Cc$ to be the set of cheap paths from $\Ff^\circ$ whose internal vertices are contained in $W_\Cc$. Claim~\ref{cl:cheap} ensures us that, indeed, each path of $\Ff^\circ$ falls into one of these sets. By Claim~\ref{cl:expensive} we have that 
\begin{equation}\label{eq:expensive}
|\Ff^\circ_\mexp|\leq k-1.
\end{equation}
On the other hand, for each $\Cc\in \Iff$ we have that 
\begin{equation}\label{eq:cheap}
|\Ff_\Cc|=\conn_{G_\Cc}(s,t)\geq |\Ff^\circ_\Cc|.
\end{equation}
By combining~\eqref{eq:expensive} and~\eqref{eq:cheap}, we infer that
\begin{equation*}
|\Ff^\apx|=\sum_{\Cc\in \Iff} |\Ff_\Cc|\geq \sum_{\Cc\in \Iff} |\Ff^{\circ}_\Cc|=|\Ff^{\circ}|-|\Ff^\circ_\mexp|\geq \conn_G(s,t)-(k-1).
\end{equation*}
\cqed\end{proof}

From Claim~\ref{cl:almost-done} it follows that in order to compute a maximum $(s,t)$-vertex flow and a minimum $(s,t)$-vertex cut in $G$, it suffices to start with the flow $\Ff^\apx$ and apply flow augmentation (Lemma~\ref{thm:augm}) at most $k-1$ times. This takes time $\Oh(k^2\cdot n)$, since by Lemma~\ref{lem:tw-edges}, $G$ has at most $kn$ edges. The algorithm clearly terminates, because the number of bags containing both $s$ ant $t$ is strictly smaller in each recursive subcall. From the description it is clear that it correctly reports a maximum $(s,t)$-vertex flow and a minimum $(s,t)$-vertex cut in $G$. We are left with estimating the running time of the algorithm.

Let $\Rt$ denote the recursion tree of the algorithm --- a rooted tree with nodes labelled by instances solved in recursive subcalls, where the root corresponds to the original instance $(G,s,t,\Tt)$ and the children of each node correspond to subcalls invoked when solving the instance associated with this node. Note that the leaves of $\Rt$ correspond to instances where there is only one bag containing both $s$ and $t$. Since the number of bags containing both $s$ and $t$ is halved in each recursive subcall, we immediately obtain the following:
\begin{claim}\label{cl:depthflow}
The height of $\Rt$ is at most $\log_2 \ell$.
\end{claim}
Let $d\leq \log_2\ell$ be the height of $\Rt$, and let $L_i$ be the set of nodes of $\Rt$ contained on level $i$ (at distance $i$ from the root), for $0\leq i\leq d$. 

Similarly as in the proof of Theorem~\ref{thm:apx}, we partition the work used by the algorithm among the nodes of $\Rt$. Each node $a$ is charged by the work needed to (i) find the balanced node $x$, (ii) investigate the connected components of $\Tt-x$ and find the important ones, (iii) construct the subcalls $(G_\Cc,s,t,\Tt_\Cc)$ for $\Cc\in \Iff$, (iv) construct $\Ff^\apx$ by putting together the flows returned from subcalls, and (v) run at most $k-1$ iterations of flow augmentation to obtain a maximum flow and a minimum cut. From the description above it follows that if in the instance associated with a node $a$ the graph has $n_a$ vertices (excluding $s$ and $t$), then the total work associated with $a$ is $\Oh(k^2\cdot n_a)$. Note that here we use the fact that $n_a>0$, because in every subproblem solved by the algorithm there is at least one vertex other than $s$ and $t$. This follows from the fact that each leaf of a clean decomposition in any subcall contains a vertex that does not belong to the leaf's parent (or any other bag). 

Obviously, each subinstance solved in the recursion corresponds to some subgraph of $G$. For a level $i$, $0\leq i\leq d$, examine all the nodes $a\in L_i$. Observe that in the instances corresponding to these nodes, all the considered subgraphs share only $s$ and $t$ among the nodes of $G$. This proves that $\sum_{a\in L_i} n_a \leq n$.
Consequently, the total work associated with the nodes of $L_i$ is $\sum_{a\in L_i} \Oh(k^2\cdot n_a)\leq \Oh(k^2\cdot n)$. By Claim~\ref{cl:depthflow} the number of levels is at most $\log_2\ell$, so the total work used by the algorithm is $\Oh(k^2\cdot n\log \ell)\leq \Oh(k^2\cdot n\log n)$.
\end{proof}

As mentioned in the introduction, the single-source single-sink result of Theorem~\ref{thm:max-flow} can be easily generalized to the multiple-source multiple-sink setting. When we want to compute a maximum $(S,T)$-vertex flow together with a minimum $(S,T)$-vertex cut, it suffices to collapse the whole sets $S$ and $T$ into single vertices $s$ and $t$, and apply the algorithm to $s$ and $t$. It is easy to see that this operation increases the treewidth of the graph by at most $2$, because the new vertices can be placed in every bag of the given tree decomposition. Similarly, in the setting when the minimum cut can contain vertices of $S\cup T$, which corresponds to finding the maximum number of completely vertex-disjoint paths from $S$ to $T$, it suffices to add two new vertices $s$ and $t$, and introduce edges $(s,u)$ for all $u\in S$ and $(v,t)$ for all $T$. Again, the new vertices can be placed in every bag of the given tree decomposition, which increases its width by at most $2$.

%% file: conclusions.tex
\newcounter{qcount}

\newcommand{\question}[1]{
\smallskip 
\noindent\begin{tabular}{ p{3em}@{}p{0.89\linewidth} }
{\bf{Q\stepcounter{qcount}\theqcount.}} & #1
\end{tabular}
\smallskip
}

In this work we have laid foundations for a systematic exploration of fully-polynomial FPT algorithms on graphs of low treewidth. We gave the first such algorithms with linear or quasi-linear running time dependence on the input size for a number of important problems, which can serve as vital primitives in future works. Of particular interest is the new pivoting scheme for matrices of low treewidth, presented in Section~\ref{sec:gaussian}, and the general Divide\&Conquer approach based on pruning a balanced bag, which was used for reconstructing a maximum matching (Theorem~\ref{thm:matching-reconstruct}) and computing the maximum vertex flow (Theorem~\ref{thm:max-flow}).

We believe that this work is but a first step in a much larger program, since our results raise a large number of very concrete research questions. In order to facilitate further discussion, we would like to state some of them explicitly in this section.

In Section~\ref{sec:approx} we have designed an approximation algorithm for treewidth that yields an $\Oh(OPT)$-approximation in time $\Oh(k^7\cdot n\log n)$. We did not attempt to optimize the running time dependence on $k$; already a better analysis of the current algorithm shows that the recursion tree in fact has depth $\Oh(k\cdot \log n)$, which gives an improved running time bound of $\Oh(k^6\cdot n\log n)$. It is interesting whether this dependence on $k$ can be reduced significantly, say to $\Oh(k^3)$. However, we believe that much more relevant questions concern improving the approximation factor and removing the $\log n$ factor from the running time bound.

\smallskip 

\question{Is there an $\Oh(OPT)$-approximation algorithm for treewidth with running time $\Oh(k^3\cdot n\log n)$?}

\question{Is there an $\Oh(\log^c OPT)$-approximation algorithm for treewidth with running time $\Oh(k^d\cdot n\log n)$, for some constants $c$ and $d$?}

\question{Is there an $\Oh(OPT^c)$-approximation algorithm for treewidth with running time $\Oh(k^d\cdot n)$, for some constants $c$ and $d$?}

\smallskip 

In Section~\ref{sec:gaussian} we have presented a new pivoting scheme for Gaussian elimination that is based on the prior knowledge of a suitable decomposition of the row-column incidence graph. The scheme works well for decompositions corresponding to parameters pathwidth and tree-partition width, but for treewidth it breaks. We can remedy the situation by reducing the treewidth case to the tree-partition width case using ideas originating in the sparsification technique of Alon and Yuster~\cite{AlonY13}, but this incurs an additional factor $k$ to the running time. Finally, there has been a lot of work on improving the running time of Gaussian elimination using fast matrix multiplication; in particular, it can be performed in time $\Oh(n^\omega)$ on general graphs~\cite{bunchH74} and in time $\Oh(n^{\omega/2})$ on sparse graph classes admitting with $\Oh(\sqrt{n})$ separators~\cite{AlonY13}, like planar and $H$-minor-free graphs. When we substitute $k=n$ or $k=n^{1/2}$ in the running time of our algorithms, we fall short of these results.

\smallskip 

\question{Can a PLUQ-factorization of a given matrix be computed using $\Oh(k^2n)$ arithmetic operations also when a tree decomposition of width $k$ is given, similarly as for path and tree-partition decompositions?}

\question{Can the techniques of Hopcroft and Bunch~\cite{bunchH74} be used to obtain an $\Oh(k^c\cdot n)$-time algorithm for computing, say, the determinant of a matrix of treewidth $k$, so that the running time would match $\Oh(n^\omega)$ whenever $k=\Theta(n)$, and $\Oh(n^{\omega/2})$ whenever the matrix has a planar graph and $k=\Theta(n^{1/2})$?}

\smallskip 

In Section~\ref{sec:matching} we presented how our algebraic results can be used to obtain fully polynomial FPT algorithms for finding the size and constructing a maximum matching in a graph of low treewidth, where the running time dependence on the size of the graph is almost linear. In both cases, we needed to perform computations in a finite field of size $\poly(n)$, which resulted in an unexpected technicality: the appearance of an additional $\log n$ factor, depending on the computation model. We believe that this additional factor should not be necessary. Also, when reconstructing the matching itself, we used an additional $\Oh(k\log n)$ factor. Perhaps more importantly, we do not see how the presented technique can be extended to the weighted setting, even in the very simple case of only having weights $1$ and $2$.

\smallskip 

\question{Can one find the size of a maximum matching in a graph without the additional $\log n$ factor incurred by algebraic operations in a finite field of size $\poly(n)$?}

\question{Can one construct a maximum matching in a graph of low treewidth in the same time as for finding its size?}

\question{Can one compute a maximum matching in a weighted graph in time $\Oh(k^c\cdot n\log n)$, where $k$ is the width of a given tree decomposition and $c$ is some constant, at least for integer weights?}

\smallskip 

In Section~\ref{sec:max-flow} we showed how a Divide\&Conquer approach can be used to design algorithms for finding maximum vertex flows in low treewidth graphs with quasi-linear running time dependence on the size of the graph. Our technique is crucially based on the fact that we work with unweighted vertex flows, because we use the property that the size of a bag upper bounds the number of paths that can use any of its vertices. Hence, we do not see how our techniques can be extended to the setting with capacities on vertices, or to edge-disjoint flows. Of course, there is also a question of removing the $\log n$ factor from the running time bound.

\smallskip 

\question{Is there an algorithm for computing a maximum $(s,t)$-vertex flow in a directed graph with a tree decomposition of width $k$ that would run in time $\Oh(k^c\cdot n)$ for some constant c?}

\question{Can one compute a maximum $(s,t)$-vertex flow in a (directed) graph with capacities on vertices in time $\Oh(k^c\cdot n)$, where $k$ is the width of a given tree decomposition and $c$ is some constant?}

\question{Can one compute a maximum $(s,t)$-edge flow in a (directed) graph in time $\Oh(k^c\cdot n\log n)$, where $k$ is the width of a given tree decomposition and $c$ is some constant? What about edge capacities?}

\smallskip

Of course, one can look for other cases when developing fully polynomial FPT algorithms can lead to an improvement over the fastest known general-usage algorithms when the given instance has low treewidth. Let us propose one important example of such a question.

\smallskip

\question{Can one design an algorithm for {\sc{Linear Programming}} that would have running time $\Oh(k^c\cdot (n+m)\log (n+m))$ for some constant $c$, when the $n\times m$ matrix of the given program has treewidth at most $k$?}

\smallskip

Finally, we remark that Giannopoulou et al.~\cite{GiannopoulouMN15} proposed the following complexity formalism for fully polynomial FPT algorithms. For a polynomial function $p(n)$, we say that a parameterized problem $\Pi$ is in $\textrm{P-FPT}(p(n))$ (for {\em{polynomial-FPT}}) if it can be solved in time $\Oh(k^c\cdot p(n))$, where $k$ is the parameter and $c$ is some constant. The class $\textrm{P-FPT}(n)$ for $p(n)=n$ is called PL-FPT (for {\em{polynomial-linear FPT}}). We can also define class PQL-FPT as $\bigcup_{d\geq 1} \textrm{P-FPT}(n \log^d n)$, that is, PQL-FPT comprises problems solvable in fully polynomial FPT time where the dependence on the input size is quasi-linear. In this work we were not interested in studying any deeper complexity theory related to fully polynomial FPT algorithms, but our algorithms can be, of course, interpreted as a fundamental toolbox of positive results for PL-FPT and PQL-FPT algorithms for the treewidth parameterization. On the other hand, the results of Abboud et al.~\cite{AbboudWW15} on {\sc{Radius}} and {\sc{Diameter}} are the first attempts of building a lower bound methodology for these complexity classes. Therefore, further investigation of the complexity theory related to P-FPT classes looks like a very promising direction.